\documentclass[11pt]{article}

\usepackage{geometry}
\usepackage{amsfonts}
\usepackage{amsmath}
\usepackage{graphicx}
\usepackage{epstopdf}
\usepackage{enumitem}

\usepackage{amsthm}

\usepackage{authblk}
\usepackage{url}
\usepackage{cleveref}

\usepackage{graphicx,epstopdf}
\usepackage[caption=false]{subfig}

\usepackage{algorithm}
\usepackage{algorithmic}

\usepackage{array}
\newcolumntype{M}[1]{>{\centering\arraybackslash}m{#1}}
\newcolumntype{N}{@{}m{0pt}@{}}

\newcommand{\DATA}{ \mathcal{D}}
\newcommand{\GIVEN}{ \, | \,}
\newcommand{\REAL}{ \mathbb{R} }
\newcommand{\HESSIAN}{ \mathcal{H} }
\newcommand{\FIM}{ \mathcal{I} }
\newcommand{\COMMA}{ \; ,}
\newcommand{\PERIOD}{ \; .}
\newcommand{\NORMAL}{ \mathcal{N} }
\newcommand{\UNIFORM}{ \mathcal{U} }
\newcommand{\PARTIAL}[1]{  \frac{\partial}{\partial #1}  }

\newcommand{\MEANN}[2]{ \mathbb{E}_{ #1 } \left [  #2 \right ] }
\newcommand{\PARTIALTWO}[2]{  \frac{\partial^2}{\partial #1 \partial #2 }  }

\newcommand{\cov}{pseudo covariance}

\newcommand{\SECD}[1]{\frac{\partial^2}{\partial {#1}^2}}

\usepackage{dsfont}
\DeclareRobustCommand{\ONE}{\text{\usefont{U}{bbold}{m}{n}1}}

\newtheorem{proposition}{Proposition}

\title{Langevin Diffusion for Population Based Sampling with an Application in Bayesian Inference for Pharmacodynamics}

\author{  Georgios Arampatzis,
 Daniel W\"alchli, Panagiotis Angelikopoulos\thanks{currently at D. E. Shaw Research L.L.C, 10036 New York, NY, USA} , 
 Stephen Wu\thanks{currently at Institute of Statistical Mathematics,  190-8562 Tokyo, Japan} ,
 \\  Panagiotis Hadjidoukas and  Petros Koumoutsakos }
\affil{Professorship for Computational Science, ETH-Zurich, CH-8092, Switzerland}

\date{}

\usepackage{amsopn}

\DeclareMathOperator*{\argmax}{arg\,max}
\DeclareMathOperator*{\DIAG}{diag}

\begin{document}

\maketitle

\begin{abstract}
We  propose an algorithm for the efficient and robust sampling of the posterior probability distribution in Bayesian inference problems. The algorithm combines the local search capabilities of the Manifold Metropolis Adjusted Langevin transition kernels with the advantages of global exploration by a population based sampling algorithm, the Transitional Markov Chain Monte Carlo (TMCMC). The Langevin diffusion process is determined by either the Hessian or the Fisher Information of the target distribution with appropriate modifications for non positive definiteness. The present methods is shown to be superior over other population based algorithms, in sampling probability distributions for which gradients  are  available and is shown  to handle otherwise unidentifiable models. We demonstrate the capabilities and advantages of the method in computing the posterior distribution of the parameters in a Pharmacodynamics model, {for glioma growth and its drug induced inhibition, using clinical data.}
\end{abstract}

\section{Introduction}

The unprecedented availability of experimental and observational data and increased computational power have fueled the re-emergence of Bayesian inference \cite{Jaynes:probability} for quantifying the uncertainty in the predictions of mathematical models. {Bayesian inference is revolutionizing simulation science by integrating mathematical models with data and prior knowledge to enable robust predictions \cite{Ghosh:bayes,Jaynes:probability,OHagan2001} in data rich domains such as Medicine an Drug discovery \cite{Sherman:2017}}. 

The Bayesian framework amounts to adjusting a subjective belief of a computational  model and its parameters using data from the underlying physical process. The resulting posterior probability  can then be used to robustly quantify the uncertainty in the model predictions. The practical value and computational cost of the Bayesian framework is largely determined by the effective sampling of the associated  probability distributions. 
In the last decade a number of algorithms have been proposed to enhance this sampling by improving on the fundamental concept of Markov Chain Monte Carlo  (MCMC) \cite{Betancourt2013,Beck2013,Haario2006,Marjoram2003,Xifara2014}. 
{Of particular interest are Monte Carlo algorithms \cite{Moral2006} that may exploit massively parallel computer architectures.} {While MCMC algorithms are fundamentally  ``sequential''  several algorithms have been proposed in the last twenty years to address their parallel implementation. Among these algorithms  are  the annealed importance sampling \cite{Neal1998}, the power posterior algorithm \cite{Friel2008}, the parallel tempering algorithm \cite{Hukushima1995}, the equi-energy sampler \cite{Kou2006}, and parallel adaptive Metropolis \cite{Solonen2012}.}
The Transitional MCMC algorithm (TMCMC) \cite{Ching2007} and its {improved} version (BASIS \cite{Wu2015}), are among the most prominent population (or particle) based MCMC. They exploit information accumulated in the population to escape local modes in order to  explore effectively multi-modal or highly peaked probability distributions  \cite{Green2015}. 
{In TMCMC, a partial MCMC is performed for each member of the population, via the Metropolis-Hastings (MH) algorithm with Gaussian proposal distribution \cite{Robert:2005}. The MH algorithm with Gaussian proposal distribution and constant covariance matrix corresponds to the discretization of an isotropic diffusion in the parameter space \cite{Roberts2002}. In TMCMC the matrix in the proposal distribution is the scaled sample covariance matrix of the population.}
However, in many situations, for instance in multi-modal or close to unidentifiable posterior distributions, the assumption of an isotropic covariance matrix, inherent to TMCMC, is inadequate. We note that sampling based on a non-isotropic covariance has been reported  in the Metropolis Adjusted Langevin Algorithm (MALA) \cite{Roberts2002,Roberts1998}. In MALA, the proposal distribution is obtained from the discretization of a Langevin diffusion with a drift term related to the gradient of the target distribution and an adjustable constant diffusion coefficient. The optimal selection of the respective covariance matrix remains an open problem. In \cite{Girolami2011} the MALA algorithm is combined  with a discretization of the Langevin diffusion on a general manifold. The authors proposed a covariance matrix that is connected with the Hessian or the Fisher information of the target distribution. The use of a covariance matrix is related to the Hessian \cite{Bui-Thanh2012} pertaining to the Newton method and to population based optimization algorithms such as CMA-ES \cite{Hansen2003}. Another class of potent sampling algorithms are the Hybrid or Hamiltonian Monte Carlo (HMC) \cite{Duane1987,Girolami2011} and their extensions on a general manifold \cite{Brubaker2012}. We refer the reader to \cite{Green2015} for an extensive review  in sampling algorithms for Bayesian computations.

In this article we augment the capabilities of the TMCMC algorithm by substituting the isotropic coefficient  with a drift term and a diffusion coefficient that reflect  the local geometry of the posterior distribution. We combine the TMCMC algorithm with Langevin diffusion transition kernels by following the manifold approach presented in \cite{Girolami2011}. We explore the Hessian and/or the Fisher information of the target distribution as local metrics to construct the covariance matrix of the proposal distribution. The proposed  manifold TMCMC (mTMCMC) has been implemented for single core\footnote{Matlab code can be downloaded from \url{http://cse-lab.ethz.ch/software/smtmcmc}} and multicore clusters\footnote{$\Pi4$U can be downloaded from \url{http://www.cse-lab.ethz.ch/software/Pi4U}} \cite{Hadjidoukas2015}. 
Positive definiteness of the local metric  is a key feature of MALA. However, this property is not ensured a-priori for every probability distribution. Moreover, in the presence of non identifiable manifolds in the target distribution, the eigenvalues of the metric become arbitrarily small leading to proposal distributions with ill-conditioned covariance matrices resulting in sampling with  low acceptance rate. In this paper, the non-invertibility is handled by disregarding the metric and using the covariance matrix usually employed  by the TMCMC algorithm. The cases of non positive definiteness and non-identifiability are treated by decomposing the matrix associated with the particular metric and appropriately scaling the problematic eigenvectors. This technique is discussed in detail in  \cref{sec:covariance}.

The proposed algorithm is first tested in a collection of multivariate Gaussian distributions, to showcase the advantages of the proposed metric correction scheme.
The effectiveness of mTMCMC to sample challenging posterior distributions is further demonstrated  in the Bayesian inference of a Pharmacodynamics problem. {The  Pharmacodynamics model describes the evolution of the mean diameter of low grade gliomas \cite{Ribba2012} under different drug therapies. Clinical data} obtained from MRIs of different patients, are used to infer the model parameters. The posterior probability distribution of the parameters is not readily invertible. The TMCMC algorithm \cite{Ching2007} was unable  to reproduce the posterior probability. The presence of at least one  non-identifiable manifold in the posterior distribution is being exposed by the use of the Profile Likelihood technique \cite{Raue2012}. We find that while the TMCMC algorithm exhibited inability to sample the parameter space, the mTMCMC was well capable of exploring the parameter space.  Finally, we demonstrate the ability of mTMCMC to sample areas of high probability by comparing its results with those from CMA-ES, a state of the art population based optimization algorithm \cite{Hansen2003}. 

The paper is organized as follows: in \cref{sec:background} we present background information on Bayesian inference and sampling using TMCMC and manifold MCMC.  In \cref{sec:mtmcmc}  we present the proposed manifold TMCMC algorithm and discuss the implementation details, {and in \cref{sec:numerics} we showcase  the ability of the proposed algorithm to sample multi-modal distributions in the presence of unidentifiable manifolds  in a Pharmacodynamics model.}

\section{Background}\label{sec:background}

We begin with an overview of  Bayesian inference  with sampling by the Transitional Markov Chain Monte Carlo (TMCMC) \cite{Ching2007} and the manifold Metropolis Adjusted Langevin Algorithm (mMALA) \cite{Girolami2011}.

\subsection{Bayesian Inference}

Given a model $f(x;\varphi)$, with $x\in \REAL^{N_x}$ the input vector and $\varphi\in\REAL^{N_\varphi}$ the parameter vector and a data set $\DATA = \{  d_i \GIVEN i=1,\ldots,N_\DATA \}$ our goal is to find parameters $\varphi$ such that $f(x;\varphi)$ is a good approximation to the observations $\mathcal{D}$. Bayes theorem provides a distribution of the model parameters conditioned on the data according to
\begin{equation}
p(\varphi | \mathcal{D})  =  \frac{   p( \mathcal{D} | \varphi)   \; p(\varphi)  }{   p( \mathcal{D})  }  \PERIOD
\end{equation}
The prior distribution $p(\varphi)$ encodes the available information on the parameters prior to observing any data. The denominator, $p(\DATA)$, is  referred to as the evidence of the data and used for model selection \cite{Beck2004}. Under the assumption that the data are independent and normally distributed around the output of the model, we postulate that
\begin{equation}
d_i = f(x_i;\varphi) + \epsilon, \quad \epsilon \sim \NORMAL(0,\sigma_n) \COMMA
\label{eq:likelihood:assumption}
\end{equation}
the likelihood function $p( \mathcal{D} | \vartheta)$ takes the form,
\begin{equation}\label{eq:model:like}
p(\DATA | \vartheta) = \NORMAL(\DATA \GIVEN F(X,\varphi), \sigma_n I)   \COMMA   
\end{equation}
where $\vartheta = ( \varphi^\top , \sigma_n)^\top$ is the parameter vector that contains both, the model and the noise, parameters and $F(X,\varphi) = (f(x_1;\varphi),\ldots,f(x_{N_\DATA};\varphi)   )$. The assumption that the observations are independent is used here in order to simplify the presentation and correlations between the observations  can be included without changing the general formulation presented here.

\subsection{Transitional Markov Chain Monte Carlo (TMCMC)}
The TMCMC is a population based algorithm  for sampling from a sequence of intermediate distributions controlled by  the annealing scheme
\begin{equation}
p_j(\vartheta) \propto p(\DATA | \vartheta)^{\zeta_j} p(\vartheta) \COMMA
\end{equation}
for $j=1,\ldots,m$ and $0=\zeta_1 < \ldots < \zeta_m = 1$, that converges to the posterior distribution $p(\vartheta|\DATA) \propto p(\DATA | \vartheta) p(\vartheta)$ when $\zeta_m = 1$.

The algorithm first  draws $N_1$ samples from the prior distribution. At the $j+1$ stage it uses $N_j$ samples from the distribution  $p_j$ to obtain $N_{j+1}$ samples from the distribution $p_{j+1}$. Let $\Theta_j=\{ \vartheta_{j,k} | k=1,\ldots,N_j \}$ be the samples obtained at the $j$-th step from $p_j$. The following procedure gives samples from $p_{j+1}$:

\begin{enumerate}
\item Draw $N_{j+1}$ samples from the set $\Theta_j$ with probability of the sample $\vartheta_{j,k}$ to be selected equal to 
\begin{equation}
\hat w_{j,k} = \frac{w_{j,k}}{\sum_{k=1}^{N_i} w_{j,k}} \COMMA
\end{equation}
where $w_{j,k} = p(\DATA | \vartheta_{j,k})^{\zeta_{j+1} - \zeta_j}$. Put the new samples in the set $\tilde \Theta_{j+1}$ and set 
\begin{equation}
S_j = \frac{1}{N_j} \sum_{k=1}^{N_j} w_{j,k} \PERIOD
\end{equation}

\item For each sample in $\tilde \Theta_{j+1}$ perform MCMC with Gaussian proposal distribution and covariance matrix $\varepsilon^2\Sigma_{s}^{(j)}$. Here, $\varepsilon$ is a scaling parameter and $\Sigma_s^{(j)}$ is the sample covariance at the $j$-th stage given by,
\begin{equation}\label{eq:weighted:cov}
\Sigma_s^{(j)} = \sum_{k=1}^{N_j} \hat w_{j,k} (\vartheta_{j,k} - \bar \vartheta_j) (\vartheta_{j,k} - \bar \vartheta_j)^\top \COMMA
\end{equation}
where $\bar \vartheta_j = \sum_{k=1}^{N_j} \hat w_{j,k}  \vartheta_{j,k}$.
Set the chain length equal to a predefined parameter $\ell_{max}$.

\end{enumerate}
The pseudocode of the algorithm is presented in  \cref{alg:tmcmc}. In \cite{Wu2015} it was suggested that the MCMC step, described in lines \ref{alg:tmcmc:line:mcmc}-\ref{alg:tmcmc:line:mcmc:end} of  \cref{alg:tmcmc}, results in a bias accumulated in each stage. In the original TMCMC the MCMC is performed for each \textit{unique} sample in $\tilde\Theta_{j+1}$ with chain length equal to number of occurrences of the sample. In \cite{Wu2015} it is shown that in order for the bias to be reduced all samples in  $\tilde\Theta_{j+1}$ should perform an MCMC step with chain length equal to a parameter $\ell_{max}$, which is usually set to 1. The improved algorithm is called BASIS and an efficient implementation can be found in the $\Pi$4U framework \cite{Hadjidoukas2015}.

\begin{algorithm}
\caption{\label{alg:tmcmc} BASIS (TMCMC)}
\begin{algorithmic}[1]
\STATE \textbf{Input:} Likelihood function $p(\DATA | \vartheta)$, prior distribution $p(\vartheta)$
\item[]  \hspace{32pt} $N_j, N_{max}$ -- number of samples per stage, maximum number of stages
\item[]  \hspace{32pt} $\gamma, \varepsilon$ -- threshold parameter, scale parameter
\item[]
\STATE \textbf{Output:} $\Theta_{final}$ -- a set of samples from $p(\vartheta | \DATA)$
\item[]  \hspace{42pt}  S -- estimation for the evidence $p(\DATA)$
\item[]
\STATE Draw initial sample set $\Theta_1=\{ \vartheta_{1,k} | k=1,\ldots,N_1 \}$ from prior
\STATE Initialize $j\gets 1$, $\zeta_1\gets 0$, $S\gets 1$
\REPEAT
\STATE Choose $\zeta_{j+1}$ such that {the coefficient of variation} of $w_{j,k}<\gamma$ and $\zeta_{j+1}\leq 1$ \label{alg:tmcmc:line:schedule}
\STATE Calculate $w_{j,k}$ with the chosen $\zeta_{j+1}$
\STATE $S\gets S \cdot \frac{1}{N_j} \sum_{j=1}^{N_j} w_{j,k}$
\STATE  Obtain $\tilde\Theta_{j+1}$ by drawing $N_{j+1}$ samples from the set $\Theta_{j}$ with probabilities $\propto w_{j,k}$ 

\STATE Set $\Sigma$ the weighted covariance given by \cref{eq:weighted:cov}

\FOR{ each sample in $\tilde\Theta_{j+1}$} \label{alg:tmcmc:line:mcmc}
\STATE Perform MCMC with length equal to $\ell_{max}$ \label{BASISvsTMCMC} and proposal distribution $q(\cdot | \vartheta )=\NORMAL(\cdot | \vartheta, \varepsilon^2\Sigma)$
\STATE Add resulting samples in $\Theta_{j+1}$
\ENDFOR \label{alg:tmcmc:line:mcmc:end}

\STATE $j \gets j+1$

\UNTIL{$\zeta_j=1$ or $j>N_{max}$}

\STATE $\Theta_{final} \gets \Theta_{j}$

\end{algorithmic}
\end{algorithm}

A key advantage of TMCMC is that it can be efficiently  parallelized since the likelihood evaluation is independent for each sample. An additional computational benefit introduced by the BASIS algorithm is that all MCMC chains have equal length and thus the work load can be balanced among the processors. 
Moreover, an important byproduct of the algorithm is that the evidence of the data is {estimated} by
\begin{equation}\label{eq:evidence}
p(\DATA) \approx \prod_{j=1}^{m-1} S_j \COMMA
\end{equation}
where $m$ is the total number of stages in the TMCMC algorithm \cite{Ching2007}. {In fact \cref{eq:evidence} is an unbiased estimator of the evidence \cite{Moral2006}}.
The idea of {estimating} the evidence of the data by \cref{eq:evidence} can also be found in  \cite{Calderhead2009,Lartillot2006} in the context of thermodynamic integration. The main difference  between these algorithms and TMCMC is that in TMCMC the annealing schedule $\{\zeta_j | j=1,\ldots,m\}$, is estimated adaptively according to the scheme described in line \ref{alg:tmcmc:line:schedule} of  \cref{alg:tmcmc}. In the thermodynamic integration the annealing schedule is chosen a priori.

Finally, we note that the parameter $\varepsilon^2$ was proposed in \cite{Ching2007} to be set equal to $0.04$. In  \cite{Betz2016} the authors propose an adaptive choice of $\varepsilon^2$ such that a predefined acceptance rate is achieved. We return to the issue of choosing $\varepsilon^2$ in  \cref{sec:sampling:schemes}.


\subsection{Manifold Metropolis Adjusted Langevin Algorithms}\label{sec:mMALA}

Let $\pi:\Theta\rightarrow\REAL$ be a probability distribution function where $\Theta\subset\REAL^{N_{\vartheta}}$. The Metropolis-Hastings (MH) algorithm is a Markov Chain Monte Carlo (MCMC) method for obtaining samples from $\pi$ according to the following iterative scheme, starting from sample $\vartheta_0$:
\begin{enumerate}
\item propose a sample $\vartheta^\star$ according to a probability function $q(\vartheta^\star|\vartheta_j)$,
\item set $\vartheta_{j+1}=\vartheta^\star$ with probability 
\begin{equation}\label{eq:acceptance:ratio}
\alpha(\vartheta^\star | \vartheta_j) = \min \left(   \,  1  \,  , \, \frac{\pi(\vartheta^\star)q(\vartheta_j|\vartheta^\star)}{\pi(\vartheta_j)q(\vartheta^\star|\vartheta_j)}  \,  \right ) \COMMA
\end{equation}
 and $\vartheta_{j+1}=\vartheta_j$ with probability $1-\alpha(\vartheta^\star | \vartheta_j)$.
\end{enumerate}
In the limit, the samples obtained using the MH algorithm will be distributed according to $\pi$. The proposal distribution $q$ is usually chosen to be a Gaussian distribution centered at $\vartheta_j$ with covariance matrix $\sigma^2I$, i.e., $q(\cdot |\vartheta_j) = \NORMAL(\cdot | \vartheta_j,\sigma^2 I)$ and $I$ the identity matrix in $\REAL^{N_\vartheta}$. The parameter $\sigma$ needs to be tuned depending on the probability $p$; if $\sigma$ is too small the proposals will be local and the chain will not be able to efficiently explore the parameter space, if $\sigma$ is too big there will many rejected samples leading to slow convergence.

An improved proposal scheme  is based on the observation that the random variable that satisfies the stochastic differential equation (SDE),
\begin{equation}\label{eq:MALA}
d\vartheta_t  =  \frac{1}{2}\nabla \log \pi(\vartheta)\,dt + dW_t \COMMA 
\end{equation}
where $W_t$ an $N_\vartheta$ dimensional Wiener process, has $\pi$ as stationary distribution. Then, the Euler-Maruyama discretization is given by
\begin{equation}\label{eq:discr:MALA}
\vartheta_{n+1} =  \vartheta_{n} + \frac{\varepsilon}{2}  \, \nabla  \log \pi(\vartheta_n) +  \sqrt{\varepsilon}W_n,  \quad W_n\sim\NORMAL(0,I)\COMMA
\end{equation}
where $\varepsilon$ is the time step of the discretization.  Since $\varepsilon$ introduces error, $\pi$ is not anymore the equilibrium distribution of $\vartheta_n$ and thus cannot be sampled directly by solving \cref{eq:MALA}. Instead, the sample $\vartheta_{n+1}$ is used as a proposal in the MH algorithm. In other words, the proposal distribution in the MH algorithm is given by
\begin{equation}
q(\cdot |\vartheta_n) = \NORMAL( \cdot  \GIVEN   \vartheta_{n} + \frac{\varepsilon}{2}  \, \nabla  \log \pi(\vartheta_n), \varepsilon I) \PERIOD
\end{equation}
This scheme is known as the Metropolis Adjusted Langevin Algorithm (MALA). Notice that the original MH algorithm with Gaussian proposal distribution corresponds to the MALA algorithm for the SDE
\begin{equation}
d\vartheta_t  =  \varepsilon \,  dW_t \COMMA 
\end{equation}
which corresponds to an isotropic diffusion in $\REAL^{N_\vartheta}$ and is usually referred to as the random walk MH algorithm (RWMH).

Although the proposals based on \cref{eq:MALA} follow the direction of the maximum change, guiding $\vartheta$ to regions of high probability, the isotropic diffusion may be inappropriate in the presence of highly correlated random variables. The following SDE offers a better proposal scheme,
\begin{equation}\label{eq:precMALA}
d\vartheta_t  =  \frac{1}{2}\Sigma \nabla \log \pi(\vartheta)\,dt + \sqrt{\Sigma} dW_t  \COMMA 
\end{equation}
where the correlation of the variables is encoded in the constant and positive definite matrix $\Sigma$. The algorithm is known as the pre-conditioned MALA \cite{Roberts2002}.
A variation of the pre-conditioned MALA \cite{Girolami2011} is based on the following SDE with position dependent covariance matrix
\begin{equation}\label{eq:manifold:diffusion}
d \vartheta_t = \frac{1}{2} G^{-1}(\vartheta_t) \nabla \log \pi(\vartheta_t)dt +  \Omega(\vartheta_t) dt + G^{-\frac{1}{2}}(\vartheta_t) dW_t 	\COMMA
\end{equation}
where 
\begin{equation}\label{eq:Omega:a}
\Omega_i(\vartheta_t) = |G(\vartheta_t)|^{-\frac{1}{2}} \sum_j \frac{\partial}{\partial \vartheta_j} \Big [ G_{i,j}^{-1}(\vartheta_t)  |G(\vartheta_t)|^{\frac{1}{2}} \Big ] \COMMA
\end{equation}
and $G$ is a positive definite matrix. The \cref{eq:manifold:diffusion} describes  diffusion in a manifold, defined in local coordinates by $G$ \cite{Girolami2011} and the resulting MCMC algorithm is called manifold MALA (mMALA).  In \cite{Xifara2014} it is shown that the correct form of $\Omega$ should be
\begin{equation}\label{eq:Omega:b}  
\Omega_i(\vartheta_t)  = \frac{1}{2} \sum_{j} \PARTIAL{\vartheta_j} G^{-1}_{i,j}(\vartheta_t) =  
{ -\frac{1}{2} \sum_{j}  \Big [  G^{-1}(\vartheta_t) \frac{ \partial G(\vartheta_t )  }{\partial \vartheta_j}  G^{-1}(\vartheta_t)  \Big ]_{i,j} }
  \COMMA
\end{equation}
where \cref{eq:Omega:a} and \cref{eq:Omega:b} describe equivalent diffusions under the condition $\partial_{\vartheta_j} G_{k,m}(\vartheta) = \partial_{\vartheta_k} G_{j,m}(\vartheta)$. The authors refer to the resulting sampling scheme as position-dependent MALA (pMALA). A simplified version,  under the assumption that the manifold has constant curvature, can be obtained by setting $\Omega=0$ leading to proposals,
\begin{equation} \label{eq:smtmcmc:a}
q(\cdot |\vartheta) = \NORMAL\big( \cdot \GIVEN  \vartheta + \frac{\varepsilon}{2}  G^{-1}(\vartheta) \, \nabla  \log \pi(\vartheta) \, , \, \varepsilon  G^{-1}(\vartheta) \, \big) \PERIOD
\end{equation}
The resulting scheme is known as simplified manifold MALA (smMALA) \cite{Girolami2011}.

However, there are two remaining questions:
\begin{enumerate}
\item{} what is the optimal choice for $G$? and 
\item{} how the parameter $\varepsilon$ should be chosen?
\end{enumerate}

An optimal scaling of $\varepsilon$  with respect to the dimension of $\vartheta$ for various MH schemes (MALA included) has been derived \cite{Roberts2002} while an 
optimal scaling for a class of mMALA algorithms has been provided for a wide class of distributions \cite{Bui-Thanh2012}. We return to this issue in  \cref{sec:sampling:schemes} where we discuss a heuristic procedure for the automatic tuning of $\varepsilon$ in the framework of TMCMC.

In order to answer the first question we consider a specific form for $\pi$ as a posterior of the Bayesian inference problem, i.e., $\pi(\vartheta) = p(\vartheta | \DATA)  \propto p(\DATA | \vartheta)  p(\vartheta)$. Then there are two widely used choices: the negative of the Hessian of $\log p(\DATA , \vartheta)$,
\begin{equation}\label{eq:hessian}
\begin{split}
G(\vartheta) = \mathcal{H}(\vartheta) :=& -  { \SECD{\vartheta} }    \log p( \DATA, \vartheta) \\
=& - { \SECD{\vartheta} }  \log p( \DATA | \vartheta) - { \SECD{\vartheta} }  \log p(\vartheta) \COMMA
\end{split}
\end{equation}
and the Fisher information of $p(\DATA , \vartheta)$,
\begin{equation}\label{eq:fim}
\begin{split}
G(\vartheta) = \mathcal{I}(\vartheta) :=&  -  \MEANN{\DATA | \vartheta}{  \SECD{\vartheta}  \log p(\DATA , \vartheta) }  \\
=&   -  \MEANN{\DATA | \vartheta}{  \SECD{\vartheta}  \log p(\DATA | \vartheta) }    -  \SECD{\vartheta}  \log p(\vartheta)    \COMMA
\end{split}
\end{equation}
which is the Fisher information of the likelihood function minus the Hessian of the prior distribution. 
In  \cref{sec:covariance} we discuss the most suitable choice of $G$.

{
Finally, in \cite{Martin2012} a similar proposal distribution to \cref{eq:smtmcmc:a} has been proposed, where  a new point is proposed according to,
\begin{equation}
\vartheta_{n+1} = \vartheta_{n} +  \mathcal{H}^{-1}(\vartheta_n) \nabla \log \pi(\vartheta_n) + \mathcal{H}^{-\frac{1}{2}}(\vartheta_n) W_n,  \quad W_n\sim\NORMAL(0,I) \PERIOD
\end{equation}
The method is called Stochastic Newton MCMC and compared to  \cref{eq:smtmcmc:a} $G=\mathcal{H}$, $\varepsilon=1$ and the $\frac{1}{2}$ factor has been dropped.
}


\section{Manifold Transitional Markov Chain Monte Carlo}\label{sec:mtmcmc}

In the TMCMC algorithm, the proposal distribution at the $j$-th stage at point $\vartheta$ is a Gaussian density function centered at $\vartheta$ with a constant covariance matrix, i.e., $\vartheta' \sim \NORMAL(\vartheta,\varepsilon^2\Sigma_s^{(j)})$,
where $\Sigma_s^{(j)}$ is the weighted sample covariance matrix at stage $j$ given by \cref{eq:weighted:cov}. 

Here, we  improve the quality of TMCMC samples by using a proposal scheme based on the manifold MALA algorithms discussed in  \cref{sec:mMALA}.

In the $j$-th stage of TMCMC samples from the distribution
\begin{equation}
p_j(\vartheta) = \frac{p(\DATA | \vartheta)^{\gamma_j} p(\vartheta)}{p_j(\DATA)} \COMMA
\end{equation}
must be collected based on some proposal.
In order to use the proposal scheme \cref{eq:manifold:diffusion} the gradient of $\log p_j$ should be computed,
\begin{equation}\label{eq:gradient:fj}
\nabla \log p_j(\vartheta) =  {\gamma_j} \nabla \log p(\DATA | \vartheta) +  \nabla \log p(\vartheta) \PERIOD
\end{equation}
{For the rest of the presentation we restrict the discussion in uniform prior distributions, thus the last term in \cref{eq:gradient:fj}, as well as the last term in $G(\vartheta)$ in \cref{eq:hessian} and \cref{eq:fim}, vanishes.}
In this case the diffusion \cref{eq:manifold:diffusion} is written as,

\begin{equation}\label{eq:manifold:diffusion:gamma}
d \vartheta_t = \frac{\gamma}{2} \Sigma_\gamma \nabla \log p(\DATA |\vartheta_t) dt +  \Omega_\gamma (\vartheta_t) dt + \sqrt{\Sigma_\gamma(\vartheta_t)} dW_t \COMMA
\end{equation}
where $\Sigma_\gamma=G_\gamma^{-1}$  with  $G_\gamma$  the metric corresponding to $p^\gamma(\DATA | \vartheta)$ and $\Omega_\gamma$ is given by \cref{eq:Omega:b} with $G$ substituted by $G_\gamma$.
{In general, the matrix $G_\gamma$ may not be invertible or positive definite. We define as $\Sigma_\gamma$ the \textit{\cov} matrix and we discuss this issue in more details in Section \ref{sec:covariance}.}
Note, that under the modeling assumption \cref{eq:model:like} and with  $G$ given by \cref{eq:hessian} and \cref{eq:fim} it can be shown that,
\begin{equation}
\Sigma_\gamma(\vartheta) = \gamma^{-1} \Sigma(\vartheta) \quad  \textrm{  and  } \quad  \Omega_\gamma(\vartheta) = \gamma^{-1} \Omega(\vartheta) \PERIOD
\end{equation}
Moreover, the gradient of the log-likelihood function is given by,
\begin{equation}
\PARTIAL{\vartheta_k}  \log p(\DATA | \vartheta) =
\begin{cases}
\sigma_n^{-2} \sum_{i=1}^{N_d} \big( d_i - f_i \big ) \PARTIAL{\varphi_k} f_i, & k=1,\ldots,N_{\vartheta}-1   \COMMA \\
-N_{d}\sigma_n^{-1} + \sigma_n^{-3}   \sum_{i=1}^{N_d} \big( d_i - f_i \big )^2 , & k=N_{\vartheta} \COMMA
\end{cases}
\label{eq:loglik:d1}
\end{equation}
where $f_i = f(x_i;\varphi)$ and the Hessian  is given by,
\begin{equation}
\PARTIALTWO{\vartheta_k}{\vartheta_\ell}  \log p(\DATA | \vartheta) =
\begin{cases}
\sigma_n^{-2} \sum_{i=1}^{N_d} \big( d_i - f_i \big ) \big (  \PARTIAL{\varphi_k \varphi_\ell}  f_i -     \PARTIAL{\varphi_\ell} f_i   \PARTIAL{\varphi_k} f_i \big), \\
\hspace{4cm} k,\ell=1,\ldots,N_{\vartheta}-1   \COMMA \\
-2 \sigma_n^{-3} \sum_{i=1}^{N_d} \big( d_i - f_i \big ) \PARTIAL{\varphi_k} f_i,  \\
\hspace{4cm} k=1,\ldots,N_{\vartheta}-1, \; \ell = N_\vartheta   \COMMA \\
-N_{d}  -  3 \sigma_n^{-4}   \sum_{i=1}^{N_d} \big( d_i - f_i \big )^2 , \\
\hspace{4cm} k,\ell=N_{\vartheta} \PERIOD 
\end{cases}
\label{eq:loglik:d2}
\end{equation}
The computation of the gradient and the Hessian of $p(\DATA | \vartheta)$ involves computation of the derivatives of the observable function $f$. We discuss the details of this computation in \cref{sec:ODEs}.

\begin{figure}
\centering
	\includegraphics[width=0.6\textwidth]{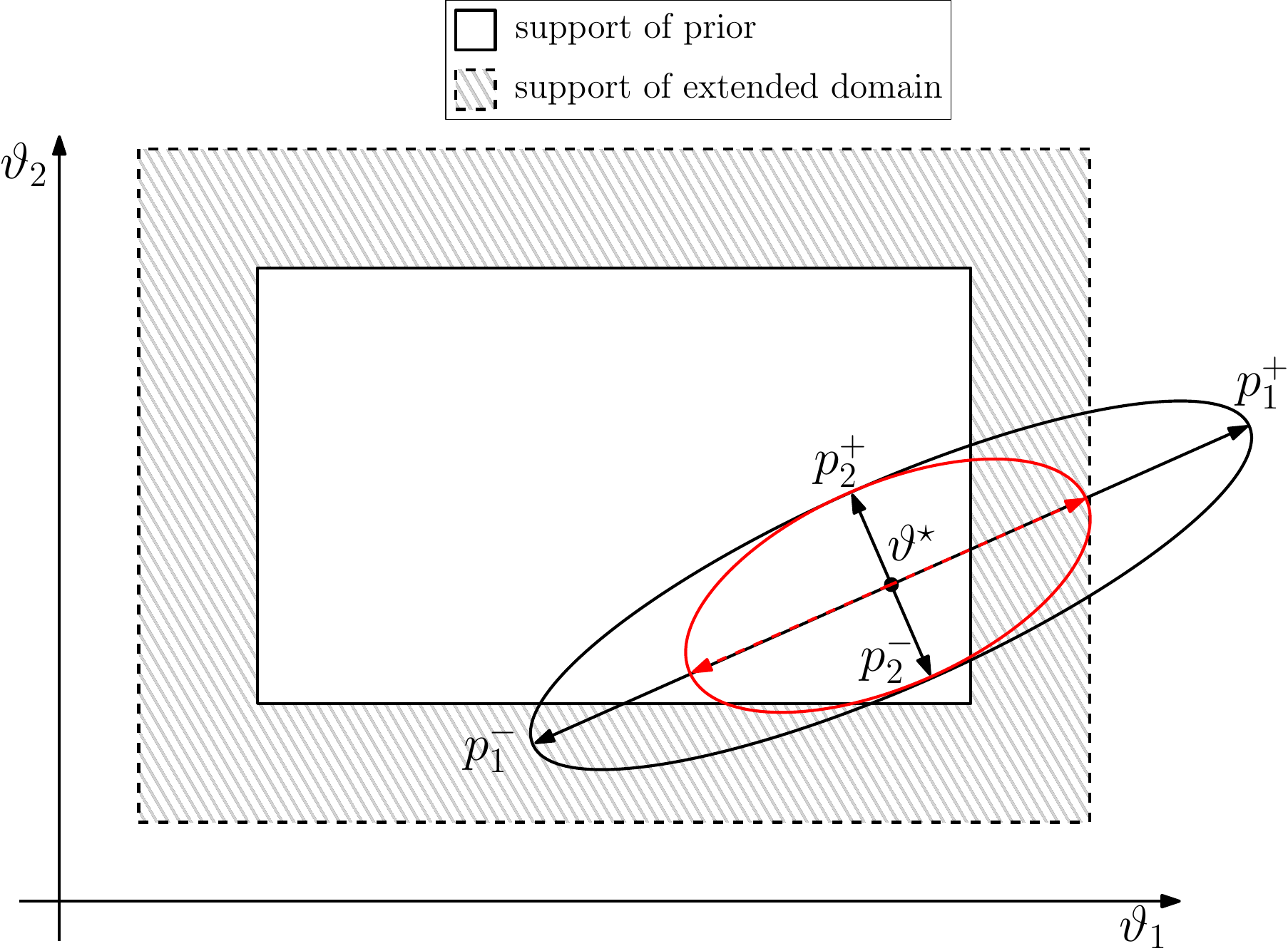}
\caption{
{ {Illustration of the eigenvalue adaptation using the extended boundary in order to treat large eigenverctors of the covariance matrix.}}
}
\label{fig:fix:eig}
\end{figure}


\subsection{Choices and corrections for the pseudo covariance  matrix}\label{sec:covariance}

We consider two possible choices of the metric $G$, the Hessian and the Fisher information, defined in \cref{eq:hessian} and \cref{eq:fim}, respectively.
These choices are associated with three problems:
\begin{enumerate} [label=(\alph*)]
\item $G$ is not invertible,
\item $G$ is invertible with some negative eigenvalues, thus $G^{-1}$  {is not positive definite},
\item {$G$ is invertible and $G^{-1}$ is positive definite but some eigenvalues of $G$ are very small, respectively some eigenvalues of $G^{-1}$ are very large. This results in a Gaussian proposal distribution that lies largely outside the bounds specified in the prior distribution.}
\end{enumerate}

The first problem is addressed by setting $G = \Sigma_s^{(j)}$, the sample covariance matrix at the $j$-th stage of TMCMC (see eq.~\ref{eq:weighted:cov}). The second problem emerges only when $G=\HESSIAN$ since the Fisher information is always positive semi-definite.  The presence of {close to} zero eigenvalues is treated in the third case. One possible solution  to this is to set $G = \Sigma_s^{(j)}$ as in the non-invertible case.
 {This fix  disregards all the information that is contained in the directions with negative eigenvalues.}
 Past proposals for fixing the {non-positive} definiteness  of the covariance matrix include the SoftAbs method  \cite{Betancourt2013} and a modified Cholesky decomposition \cite{Kleppe2015}. 
Here,  {we follow the approach used in \cite{Martin2012}} and substitute the negative eigenvalues of $G^{-1}$ with a predefined positive number chosen to be equal to the smallest eigenvalue of  {$\Sigma_s^{(j)}$}. 

{
The third problem appears in the case of unidentifiable parameters or parameter combinations. The likelihood of the data stays constant for these parameters and thus all derivatives are zero. In many practical situations, one may encounter the situation where \textit{computational unidentifiability} appears \cite{Raue2012}.
In this case the problem is identifiable in a small region of the probability space but close to unidentifiable in a large region of the parameter space. Unidentifiable directions correspond to zero, or close to zero, eigenvalues in the Hessian or the Fisher information, leading to large eigenvectors in the proposal Gaussian distribution.}
{It is important to note here that this entire discussion is particular to the case of uniform prior; for an in-depth consideration of the role of non-informed directions in the case of a Gaussian prior the reader is referred to the works \cite{Beskos2017,Cui2016,Cui2014}.}

{We  treat the problems arising from the large eigenvalues as follows: Let  $\Sigma=Q \, \Lambda \, Q^\top \in \REAL^{N_\vartheta,N_\vartheta}$ be the eigendecomposition of the covariance matrix of a normal distribution centered  at the $\vartheta^\star$, with $\Lambda=\DIAG(\lambda_1,\ldots,\lambda_{N_{\vartheta}})$, $\lambda_i \in \mathbb{R}$ the eigenvalues of $\Sigma$ and $Q$ a square matrix whose columns $q_i \in \mathbb{R}^{N_\vartheta}$ are the eigenvectors corresponding to $\lambda_i$.  Then, for $\eta\in (0,1)$ the $1-\eta$ of the probability mass lies inside the ellipsoid described by the equation
\begin{equation}\label{eq:ellipsoid}
( \vartheta-\vartheta^\star)^\top \,  Q \,   \Lambda^{-1} \, Q^\top ( \vartheta-\vartheta^\star)  = \chi_{N_\vartheta}^{2}(\eta)\, ,  \quad  {\vartheta \in \mathbb{R}^{N_\vartheta}} \COMMA
\end{equation}
where $\chi_{N_\vartheta}^{2}(\eta)$ is the upper 100$\eta$-th percentile of the  $\chi^2$ distribution with $N_{\vartheta}$ degrees of freedom \cite{Slotani1964}
The semi-axes of the ellipsoid \cref{eq:ellipsoid} are given by $\sqrt{\lambda_i \chi_k^{2}(\eta) } q_i$.
The idea is that the eigenvalues that lead to large semi-axes will be adapted such that the ellipsoid will lie inside the prior domain.  We have observed that this approach leads to small eigenvalues near the boundaries and thus the proposal distribution may be concentrated near the boundaries. 
We resolve this issue by adapting the semi-axes to an extended boundary that is defined as a percentage {$\rho\in [0,1]$} of the length of the boundaries of the prior probability.

We  adapt the large eigenvalues of $\Sigma$ by finding  constants $c_i \in \mathbb{R}$ such that the points 
\begin{equation}
\hat p_i^{\pm} = \vartheta^\star \pm \sqrt{\widehat\lambda_i \chi_k^{2}(\eta) }q_i, \qquad \widehat\lambda_i = c_i \lambda_i, \qquad i=1,\ldots,N_{\vartheta} \COMMA
\end{equation}
lie inside the extended bounds of the prior distribution. Let $\alpha,\beta\in\mathbb{R}^{N_{\vartheta}}$ be the vectors that define the uniform prior distribution, i.e., $\alpha_i \leq \vartheta_i \leq \beta_i$, and $J_{i,\alpha}^{\pm}$ and $J_{i,\beta}^{\pm}$ the sets of indices that violate the inequalities 
$$
(1-\rho)\alpha_j \leq p_{i,j}^{\pm} \quad \textrm{ and } \quad  p_{i,j}^{\pm} \leq (1+\rho)\beta_j  \COMMA
$$
respectively, where  $p_i^{\pm} = \vartheta^{\star}\pm\sqrt{\lambda_i \chi_k^{2}(\eta) } q_i, \,\, i=1,\ldots,N_{\vartheta}$. Next, we define the constants
\begin{equation}
c_{i,\alpha}^{\pm} = 
\begin{cases}
    \min \left \{  \;  \frac{1}{\lambda_i \chi_{N_\vartheta}^{2}(\eta)}  \left | \frac{ (1-\rho)\alpha_j - \vartheta_j^{\star}}{q_{i,j}}  \right |^2  \; : \; j\in  J_{i,\alpha}^{\pm}  \;   \right \}  & \quad \text{if   } J_{i,\alpha}^{\pm} \neq \emptyset \\
  \qquad \qquad\qquad\quad 1       & \quad \text{if   } J_{i,\alpha}^{\pm} = \emptyset 
\end{cases}
\end{equation}
and
\begin{equation}
c_{i,\beta}^{\pm} = 
\begin{cases}
    \min \left \{  \;  \frac{1}{\lambda_i \chi_{N_\vartheta}^{2}(\eta)} \left | \frac{ (1+\rho)\beta_j - \vartheta_j^{\star}}{q_{i,j}}  \right |^2  \; : \; j\in  J_{i,\beta}^{\pm}  \;   \right \}  & \quad \text{if   } J_{i,\beta}^{\pm} \neq \emptyset \\
  \qquad \qquad\qquad\quad 1       & \quad \text{if   } J_{i,\beta}^{\pm} = \emptyset 
\end{cases}
\end{equation}
for $i=1,\ldots,N_\vartheta$ . 
Finally, the correction constant for the $i$-th eigenvalue is given by
\begin{equation}
c_i = \min \left \{ \;   c_{i,\alpha}^{+}, \, c_{i,\alpha}^{-},  \, c_{i,\beta}^{+},  \, c_{i,\beta}^{-}        \; \right \} \PERIOD
\end{equation}

We will denote the corrected covariance by $\widehat\Sigma = Q \widehat \Lambda Q^\top$ and $\widehat \Lambda = C \Lambda$ with $C=\textrm{diag}(c_1,\ldots,c_{N_\vartheta})$. In all the numerical tests of \cref{sec:numerics} we have set $\eta=0.3$. For an illustration of this approach see \cref{fig:fix:eig}, where the eigenvalue in the direction of $p_1^+$ has been adapted such that  $p_1^+$ will lie inside the extended boundary domain. Notice that the eigenvalues in the direction of $p_2^{\pm}$ has not been changed.}

The advantages of the extended boundaries approach, as well as the selection of the parameter $\rho$, are presented in a truncated multivariate Gaussian in  \cref{sec:gaussian:truncated}.

\begin{table}[htpb]
\caption{Mean and covariance of the proposal distribution, $q(\vartheta| \vartheta^\star)  = \NORMAL \big( \vartheta \GIVEN D(\vartheta^\star;\varepsilon,\gamma),C(\vartheta^\star;\varepsilon,\gamma) \big)$ for  the TMCMC algorithms described in  \cref{sec:sampling:schemes}.}
\label{table:proposals}
\centering
\begin{tabular}{ | M{2.4cm} | M{6cm} | M{2cm}  | N} \hline
\textbf{name} 	& \textbf{mean}, $D(\vartheta^\star;\varepsilon,\gamma)$ 
			& \textbf{covariance}, $C(\vartheta^\star;\varepsilon,\gamma)$ 
			& \\[5pt] \hline
 TMCMC 	& $\displaystyle \vartheta^\star$ 
 			& $\displaystyle \varepsilon \Sigma_s^{(j)}$   
			&   \\[8pt] \hline
smTMCMC &  $\displaystyle \vartheta^\star +  \frac{\varepsilon \gamma}{2}  \widehat\Sigma_\gamma(\vartheta^\star) \nabla \log p(\DATA |\vartheta^\star) $    
		    &  $\displaystyle \varepsilon  \widehat\Sigma_\gamma(\vartheta^\star)$ 
		    &  \\[12pt] \hline
pTMCMC 	&  $\displaystyle  \vartheta^\star + \frac{\varepsilon \gamma}{2}  \widehat\Sigma_\gamma(\vartheta^\star) \nabla \log p(\DATA |\vartheta^\star) +  \varepsilon \widehat\Omega_\gamma (\vartheta^\star) $   
			&  $\displaystyle \varepsilon {\widehat\Sigma}_\gamma(\vartheta^\star)$ 
			&  \\[12pt] \hline
\end{tabular}
\end{table}

\subsection{Sampling schemes} \label{sec:sampling:schemes}
In this section we summarize the proposal schemes that we implement within the TMCMC algorithm. For the proposal distribution we use the unified notation,
\begin{equation}
q(\vartheta | \vartheta^\star) = \NORMAL \big( \vartheta \GIVEN D(\vartheta^\star;\varepsilon,\gamma),C(\vartheta^\star;\varepsilon,\gamma) \big) \COMMA
\end{equation}
with $D$ and $C$ the mean and the covariance matrix of the Gaussian proposal distribution, respectively.  In the original TMCMC we have  $D(\vartheta^\star;\varepsilon,\gamma)=\vartheta^\star$ and  $C(\vartheta^\star;\varepsilon,\gamma)=\varepsilon\Sigma_s^{(j)}$, 
{see  eq.~\cref{eq:weighted:cov}}. Since $\Sigma_s^{(j)}$ is a constant matrix the proposal distribution is symmetric and the acceptance ratio is independent of $q$.

{ The simplified manifold TMCMC (smTMCMC)  \cref{eq:smtmcmc:a}  is obtained by assuming that the metric $G_\gamma$ is locally constant.} This assumption leads to $\Sigma_\gamma(\vartheta)=\Sigma_\gamma(\vartheta^\star)$ and $\Omega_\gamma(\vartheta)=0$. The \cov matrix $\Sigma_\gamma(\vartheta)$ is substituted by $\widehat\Sigma_\gamma(\vartheta)$,  the corrected covariance according to the scheme presented in  \cref{sec:covariance}. Finally, the mean and the variance of the proposal distribution are given by   $D(\vartheta;\varepsilon,\gamma)= \vartheta +  \frac{\varepsilon \gamma}{2}  \widehat\Sigma_\gamma(\vartheta) \nabla \log p(\DATA |\vartheta) $    and    $C(\vartheta;\varepsilon,\gamma)=\varepsilon\Sigma_\gamma(\vartheta)$, respectively. Notice that the proposal distribution is not symmetric. This asymmetry  must be accounted in the calculation of the acceptance ratio (see  \cref{table:acceptance}).

{Without the assumption of a locally constant metric, the term $\Omega_\gamma$ is not necessarily zero.} After correcting the {\cov} matrix $\Sigma_\gamma$ the $\Omega$ term is written as,
\begin{equation}
\widehat\Omega_{\gamma,i}(\vartheta)  =
{
 -\frac{1}{2} \sum_{j}  \Big [   \widehat\Sigma_\gamma(\vartheta) \frac{\partial  \widehat G_{\gamma}(\vartheta)}{\partial \vartheta_j}  \  \widehat\Sigma_\gamma(\vartheta)   \Big ]_{i,j}  \COMMA
  }
\end{equation}
with $\widehat G_\gamma = \widehat\Sigma^{-1}_\gamma$. In order to express the unknown derivative of $\widehat G$ in terms of the known derivative of $G$ we follow the same procedure as in \cite{Betancourt2013}, see also  \cref{appendix:derivative:G}, 
\begin{equation}\label{eq:derivative:G}
\frac{\partial \widehat G_\gamma}{\partial \vartheta_j} = Q \Big (     J \circ      \big( Q^\top \, \frac{\partial  G_\gamma}{\partial \vartheta_j}  \, Q \big )   \Big  ) Q^\top \COMMA
\end{equation}
where $\circ$ denotes the Hadamard product and 
\begin{equation}
J_{i,j} = 
\begin{cases}
\frac{ \widehat\lambda_i - \widehat\lambda_j}{\lambda_i - \lambda_j} =   \frac{ c_i \lambda_i - c_j \lambda_j}{\lambda_i - \lambda_j} , & i\neq j \\
\frac{\partial \widehat \lambda_i}{\partial \lambda_i} = c_i, & i=j \PERIOD
\end{cases}
\end{equation}
We call the resulting algorithm position dependent TMCMC (pTMCMC).

\vspace{5pt}\noindent\textbf{Tuning the scale parameter.}
The  parameter $\varepsilon$  is tuned dependent to the sampling algorithms. A usual practice is to vary $\varepsilon$ until the acceptance rate of the algorithm reaches a desired acceptance rate. Following the results in \cite{Roberts2001} we set the target acceptance rate for TMCMC, which is based in random walk MH, equal to $0.234$ and for the Langevin diffusion based TMCMC equal to $0.574$.
It can be shown that $\varepsilon$ scales like $N_\vartheta^{-1}$ for RWMH and like $N_\vartheta^{-1/3}$ for MALA  \cite{Roberts2001} as $N_\vartheta \rightarrow\infty$.
One way to find the optimal scaling is to run some preparatory, relatively small, {TMCMC} simulations with various $\varepsilon$ and choose the one that gives the optimal acceptance rate.

An adaptive method to find the optimal scaling is proposed  in \cite{Betz2016} in the framework of TMCMC. In this method the sample sets at the $j$-th stage of TMCMC are divided into subsets. These subsets are run sequentially and information of the mean acceptance rate is being passed to the next subset. The scaling parameter is being tuned depending on the distance between the mean acceptance rate and the target acceptance rate.

\begin{table}[htpb]
\caption{ Acceptance ratio $\alpha(\vartheta | \vartheta^\star) = \min \left( 1,A(\vartheta | \vartheta^\star) \right)$ for the proposal distribution  of  the TMCMC algorithms described in  \cref{sec:sampling:schemes}.}
\label{table:acceptance}

\centering
\begin{tabular}{ | M{2.4cm}  | M{5.0cm} | N} \hline
\textbf{name} 	& \textbf{acceptance ratio}, $A(\vartheta | \vartheta^\star)$ 
			& \\[5pt] \hline
 TMCMC 	&   $\displaystyle \frac{ p(\DATA | \vartheta)  }{ p(\DATA | \vartheta^\star) } $  
			&   \\[24pt] \hline
smTMCMC 	&   $\displaystyle \frac{ p(\DATA | \vartheta)  \, \NORMAL \big(\vartheta^\star \GIVEN D(\vartheta),C(\vartheta) \big)      }{ p(\DATA | \vartheta^\star) \; { \NORMAL \big(\vartheta \GIVEN D(\vartheta^\star),C(\vartheta) \big)  }  } $
		    	&  \\[25pt] \hline
pTMCMC 		& $\displaystyle \frac{ p(\DATA | \vartheta)    \;  \NORMAL \big(\vartheta^\star \GIVEN D(\vartheta),C(\vartheta) \big)     }{ p(\DATA | \vartheta^\star) \;  \NORMAL \big(\vartheta \GIVEN D(\vartheta^\star),C(\vartheta^\star) \big)    } $
			&  \\[25pt] \hline
\end{tabular}
\end{table}

\vspace{5pt}\noindent\textbf{Choosing the metric.}
In the aforementioned sampling schemes either the Hessian or the Fisher information is used as the underlying metric. Note that the Hessian needs the computation of second order derivatives while the Fisher information requires only first order derivatives. Moreover, since the Fisher information is always positive semi-definite, we expect it to perform better than the Hessian. However, in some cases the Fisher information is not explicitly known, e.g., in the case of Gaussian mixtures, and thus the use of the Hessian is inevitable. {Finally, in the pTMCMC  scheme, where the computation of  $\widehat\Omega$  is involved, the order of the needed derivatives is increased by one. Thus, for $G=\HESSIAN$ and $G=\FIM$,  third and second order derivatives, respectively, should be computed. In order to avoid the computation of third order derivatives we choose not to implement the Hessian in this sampling scheme.}

The performance of the proposed algorithms is discussed in detail in  \cref{sec:numerics}. In  \cref{table:proposals} we present a summary of the proposal distribution and in \cref{table:acceptance} the resulting acceptance ratio for the various sampling schemes.

\section{Applications}\label{sec:numerics}

We demonstrate the capabilities of the present algorithm on a number of benchmark problems  and on a challenging Pharmacodynamics model that is calibrated using clinical data.

We first  test our algorithm in a truncated Gaussian distribution. This test demonstrates the need of the extended boundary approach discussed in  \cref{sec:covariance}. In the second example the scaling of the error of smTMCMC is compared to that of TMCMC in a multivariate Gaussian distribution.
Finally, we sample the posterior distribution of a Bayesian inference problem on the parameters of Pharmacodynamics model. The challenging part of sampling this distribution is that it has at least one large manifold of non identifiable parameters.  In this example the TMCMC algorithm completely fails to explore the parameter space while the smTMCMC algorithm provides good quality samples.
We note that in our experiments the differences between smTMCMC and pTMCMC are indistinguishable and thus only results from smTMCMC are reported. The provided code includes also the pTMCMC algorithm which employs as diffusion metric the Fisher Information.

\subsection{Truncated Gaussian distribution}\label{sec:gaussian:truncated}

\begin{figure}
\centering
        \includegraphics[width=0.8\textwidth]{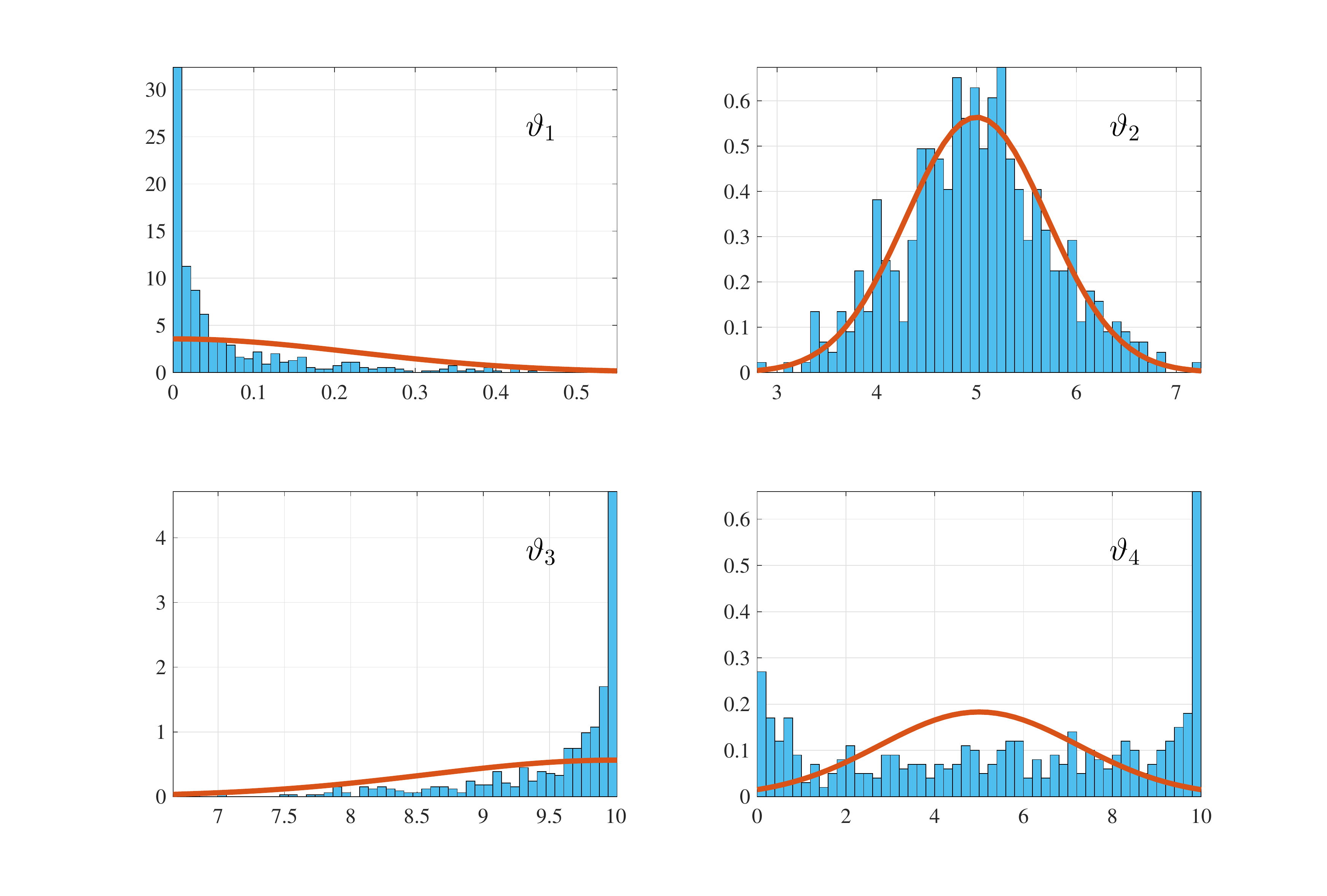}
\caption{Marginal histograms of the {sampled compared with the exact (red line) truncated Gaussian distribution \cref{eq:gaussian:uni}.} The sampling algorithm is the smMALA algorithm with no extended boundaries. Notice that the sampling algorithm incorrectly concentrates samples near the boundaries of the prior.}
\label{fig:gaussian:uni}
\end{figure}

Here, we illustrate the ability of the proposed algorithm to sample correctly distributions with mass concentrated at the boundaries of the prior distribution. The strategy of adapting the eigenvalues of the covariance matrix of the proposal distribution by extending the boundaries of the prior distribution \cref{sec:covariance} is compared with the $\rho=0$ adaptation where the eigenvalues are adapted to the boundaries of the prior. For $\rho=0$ the covariance matrix after the correction, $\widehat \Sigma$, at samples near the boundaries will have some very small eigenvalues and the MCMC chain starting at these points will be locally trapped.

We consider  the truncated Gaussian probability distribution
\begin{equation} \label{eq:gaussian:uni}
p(\vartheta) =    \prod_{i=1}^4    \NORMAL(\vartheta_i \GIVEN \mu_i,\sigma_i^2) \,  \UNIFORM(\vartheta_i \GIVEN 0,10) \COMMA
\end{equation}
with $\mu=(0,5,10,9)$ and $\sigma^2=(0.05,0.5,2,5)$. The marginal distribution of $\vartheta_i$ is depicted with a solid line in  \cref{fig:gaussian:uni}. The histograms in  \cref{fig:gaussian:uni}  are obtained with the smMALA algorithm and $G$ is the Fisher information matrix of \cref{eq:gaussian:uni} using $500$ samples and $\varepsilon=1$. It is evident from this example that the samples are concentrated at the boundaries of the prior distribution where there is significant mass of the Gaussian distribution. Notice that the variable $\vartheta_2$, which is distributed  away from the boundaries, is sampled correctly.

Next, we sample \cref{eq:gaussian:uni} using smMALA with extended boundaries and we study the effect of the parameter $\rho$ to the accuracy of the sampling. 
{Let $\tilde p$ be the estimated probability distribution using a sampling algorithm and $p$ the target distribution \cref{eq:gaussian:uni}.}
We measure the accuracy using the  relative entropy or Kullback-Leibler (KL) divergence of $p$ from $\tilde p$,
\begin{equation}\label{eq:kl}
D_{KL}( \, \tilde p \, \|  \,  p \, ) = \int_{-\infty}^\infty  \tilde p(\vartheta) \log \frac{\tilde p(\vartheta)}{ p(\vartheta) } \, d\vartheta \COMMA
\end{equation}
with $p$ absolutely continuous with respect to $\tilde p$, i.e., $p(\vartheta)=0$ implies $\tilde p(\vartheta)=0$.
The KL divergence is not a metric, since it is not symmetric and does not obey the triangle inequality, but it is a measure of loss of information if $p$ is used instead of $\tilde p$.
The reason we consider  $D_{KL}( \, \tilde p \, \|  \,  p \, )$, and not $D_{KL}( \, p \, \|  \,  \tilde p \, )$, is that $\tilde p$ (the estimated probability) is always absolutely continuous with respect to $p$ (the target distribution). Since $\vartheta_i$ are independent the KL divergence reduces to the sum of KL divergence of the marginal distributions.

\begin{figure}[tbhp]
\centering 
\subfloat[]{\label{fig:uni:err:a} \includegraphics[width=0.46\textwidth]{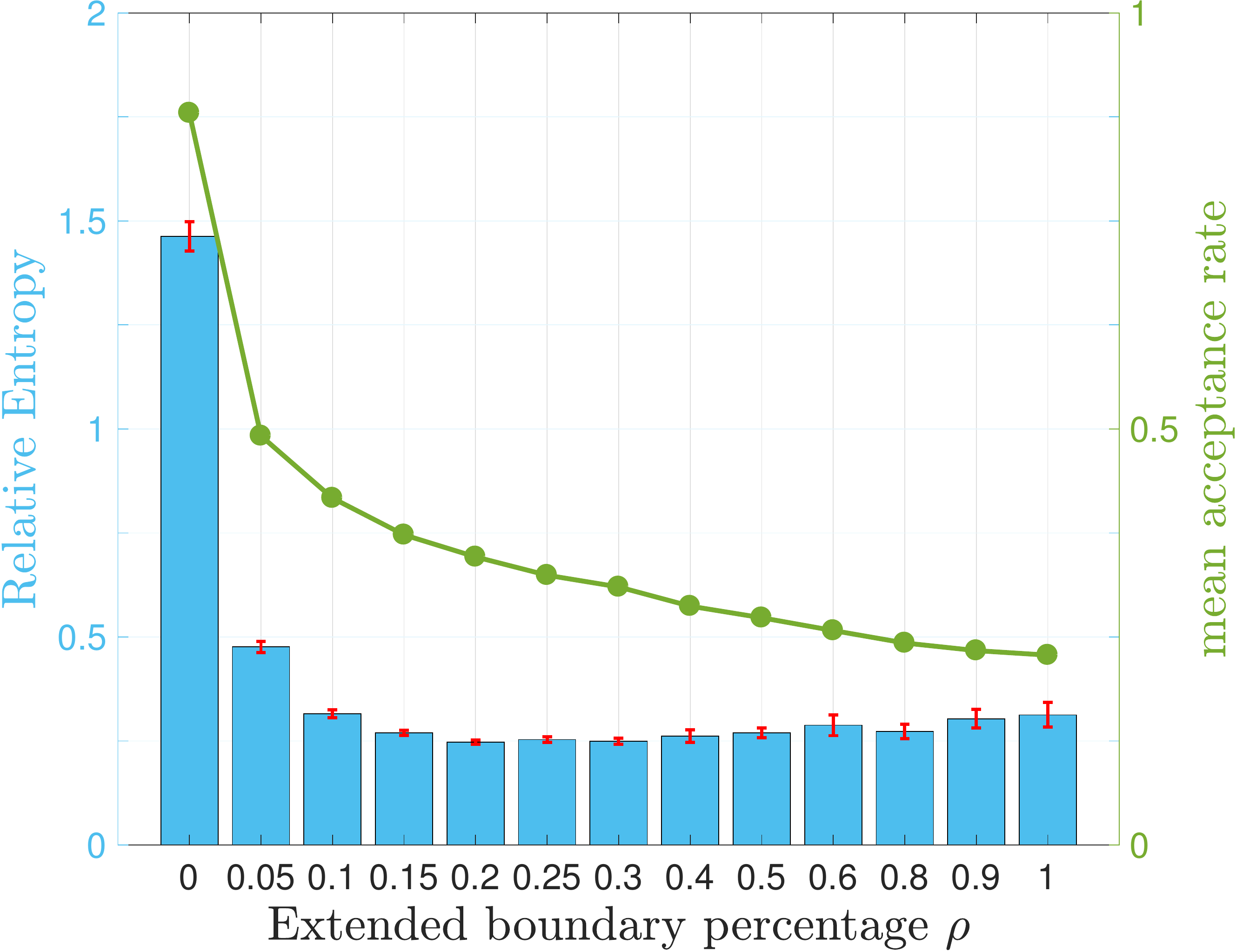} } 
\subfloat[]{\label{fig:uni:err:b} \includegraphics[width=0.44\textwidth]{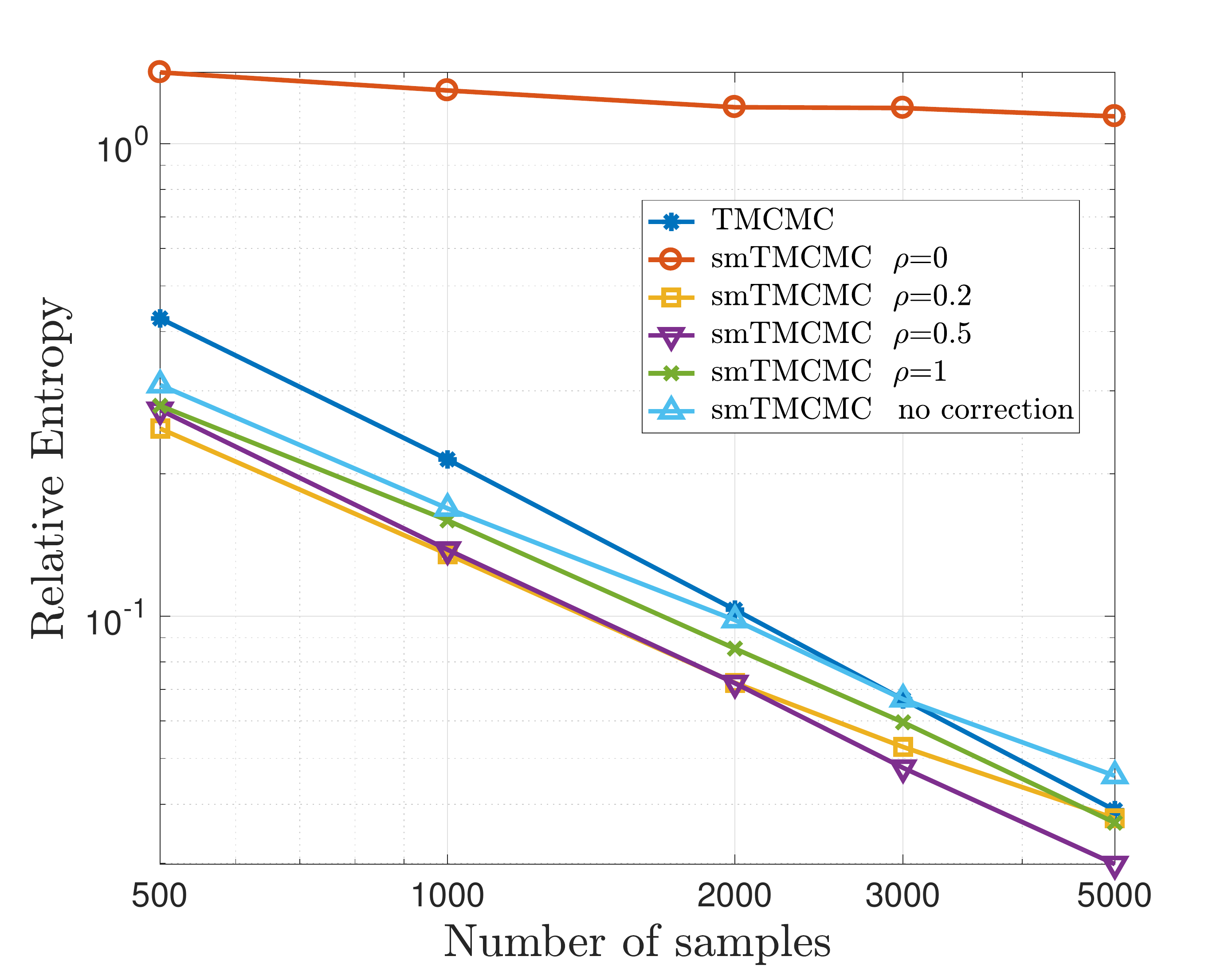} } 
\caption{(a) Estimated relative entropy (KL divergence) between the estimated and the truncated Gaussian distribution of \cref{eq:gaussian:uni} as a function of the parameter $\rho$ for the smTMCMC algorithm. (b) The same quantity as a function of the sample size of the sampling algorithm for TMCMC, smTMCMC with $\rho=0,0.2,0.5,1$ and smTMCMC with no correction.}
\label{fig:uni:err}
\end{figure}

In  \cref{fig:uni:err:a} the averaged KL divergence over $100$ independent {samplings (each with 500 samples)} is plotted as a function of $\rho$ and is depicted  with bars and a scale that corresponds to the left y-axis. 
 For $\rho=0$ the divergence reaches its maximum while for $\rho$ near $0.2$ the divergence is minimized. For $\rho$ larger than $0.3$ the divergence slightly increases but remains low.
As expected, the mean acceptance rate over the stages of smTMCMC, depicted with filled dots in  \cref{fig:uni:err:a} and scale that corresponds to the right y-axis, drops as $\rho$ increases.

In  \cref{fig:uni:err:b} the  averaged KL divergence over $200$ independent samplings is plotted as a function of the number of samples of the sampling algorithm. The KL divergence of TMCMC and smTMCMC with $\rho=0,0.2,0.5,1$ is estimated. Notice that smTMCMC with $\rho=0$ shows little improvement as the number of samples increases. This implies that the bias introduced by the eigenvalue adaptation scheme is not reduced with the size of the sample set. On the other hand, smTMCMC with $\rho=0.2,0.5$ and 1 has lower error than the TMCMC, with both slopes being equal to approximately $-1$.  {Finally, smTMCMC with no correction is shown here. As expected, the no correction scheme works well in this case as there are no non-identifiable directions. Moreover, the performance of the correction scheme is similar to that with no correction.}
Confidence intervals are not presented because they are smaller than the size of the markers in the plot.

\subsection{Multi-Dimensional Gaussian Distribution}\label{section:gaussians}

\begin{figure}[tbhp]
\centering 
\subfloat[Gaussian distribution]{\label{fig:gaussian:err:a} \includegraphics[width=0.45\textwidth]{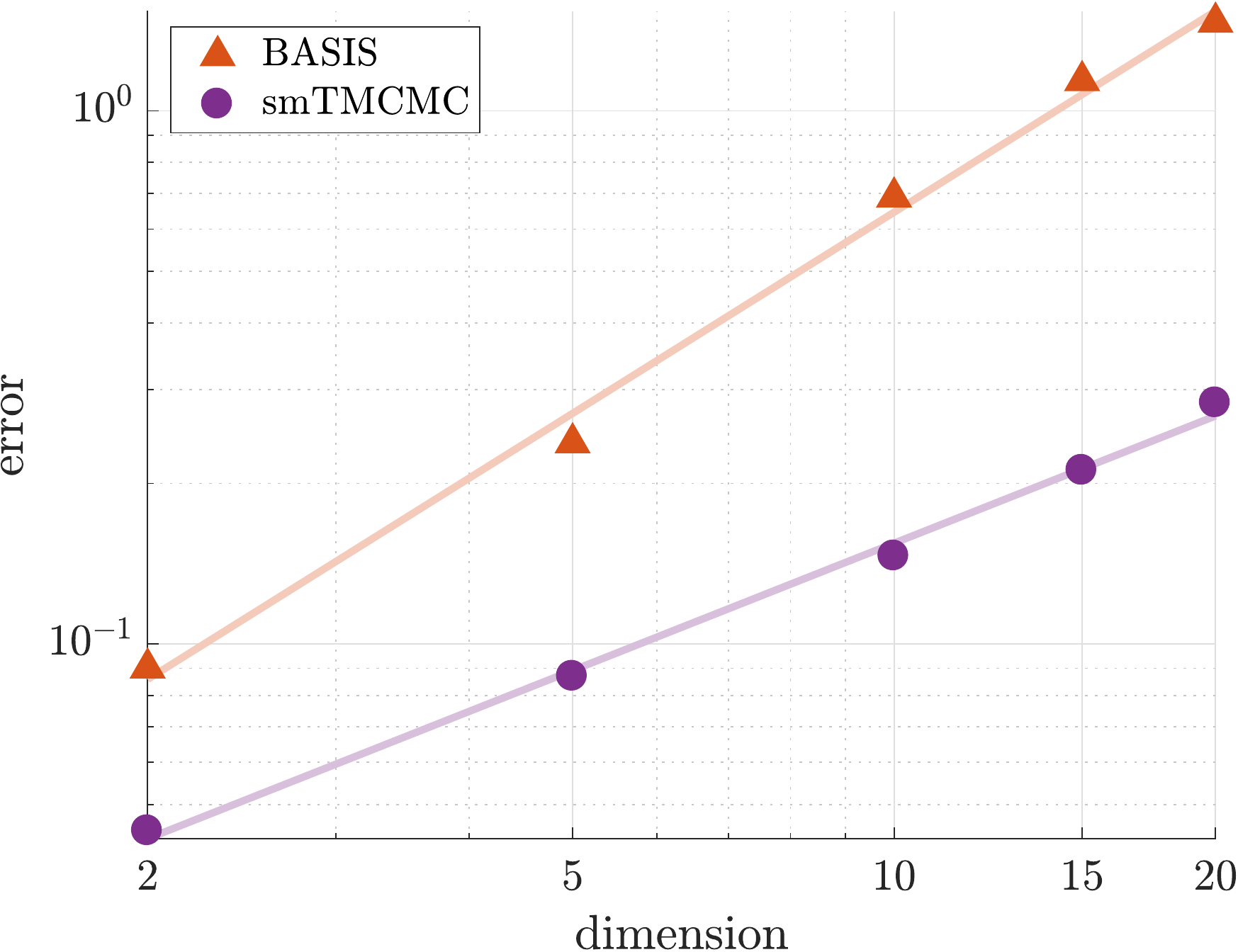} } 
\hspace{20pt}
\subfloat[Bimodal Gaussian mixture]{\label{fig:gaussian:err:b} \includegraphics[width=0.45\textwidth]{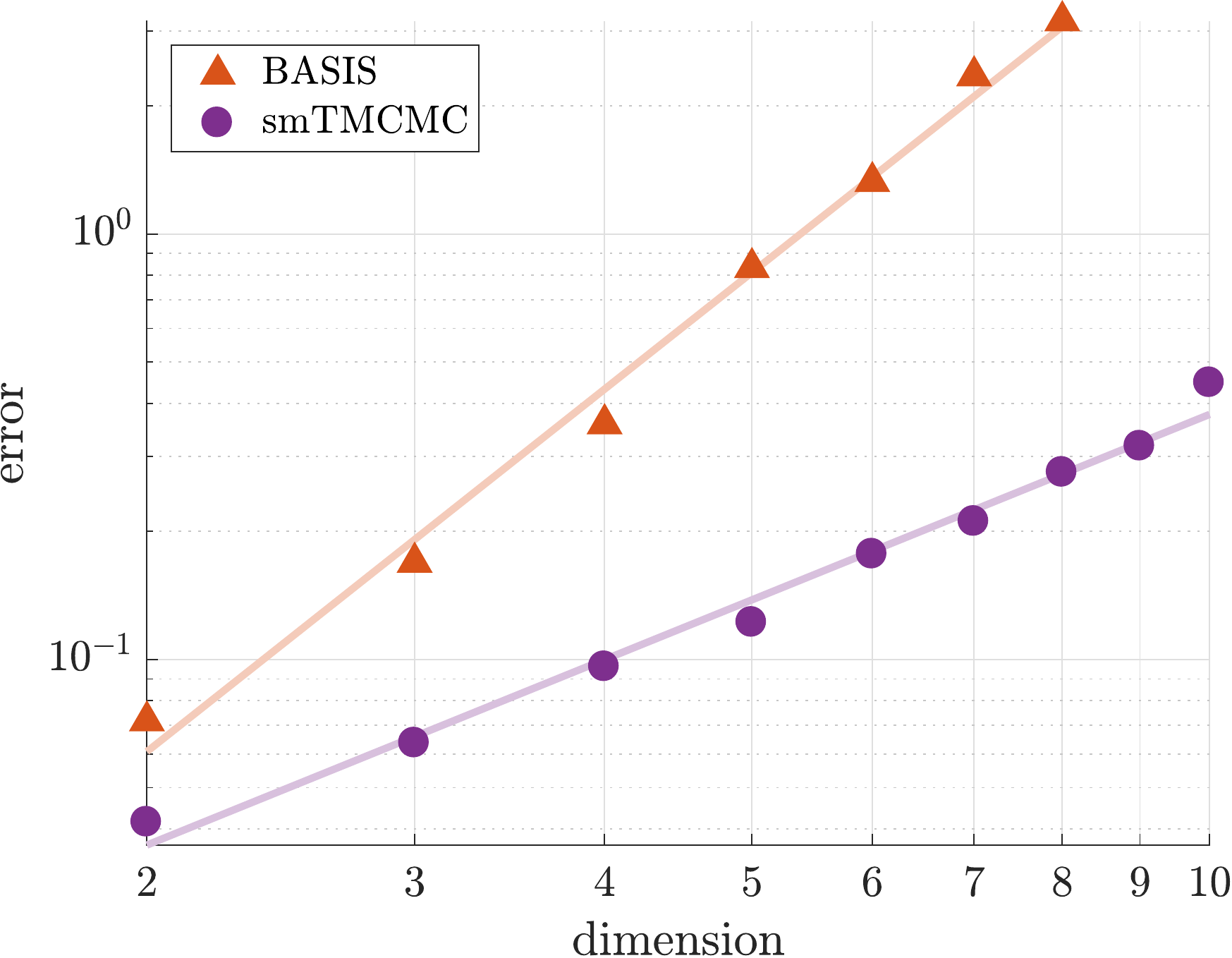} } 
\caption{Averaged error of (a) a Gaussian distribution (b) a {bimodal Gaussian mixture} distribution as a function of the dimension. {The continuous lines are the linear fits in the logarithmic scale of the data.}}
\label{fig:gaussian:err}
\end{figure}

In this section we present the scaling of the sampling error for the TMCMC and smTMCMC algorithms in a { Gaussian distribution and a bimodal Gaussian mixture}.
In the first test case, the mean $\mu$ is zero and the covariance matrix $\Sigma$ of the target distribution is randomly generated with the MATLAB function \textit{gallery('randcorr',d)}, where $d$ is the number of dimensions.
The error of the sampling algorithm is defined as $E = \frac{1}{2} (e_1+e_2)$ where \
\begin{equation}
{
e_1 = \frac{1}{d} \sum_{i=1}^d  | \bar \mu_i - \mu_i | \quad \textrm{and} \quad e_{2} = \frac{1}{d^2}  \sum_{i,j=1}^d  |\bar\Sigma_{i,j} - \Sigma_{i,j} |  \COMMA
}
\label{eq:def:err}
\end{equation}
and $\bar\mu,\, \bar\Sigma$ are the estimated mean and covariance matrix, respectively.
The reported sampling error is averaged over 100 independent simulations {using 1000 samples}.
The {initial distribution} is  a $d$-dimensional uniform distribution with support in $[-10,10]$ for each dimension. 
We set the scaling parameters $\varepsilon = 0.04$  for TMCMC and  $\varepsilon = 1$ for smTMCMC. Notice that  in this case the Hessian matrix is equal to Fisher information and equal to the inverse of the covariance matrix of the target distribution.

In  \cref{fig:gaussian:err:a} the error for $d=2,5,10,15,20$ is presented. Both errors increase with the dimension of the target distribution while the error of smTMCMC is always bellow the error of TMCMC.
{In \cref{fig:gaussian:err:vs:dim} the error for $d=5$ as a function of the number of samples is presented in logarithmic scale. The continuous lines correspond to fitted linear functions and the rate of convergence is approximately $-\frac{1}{2}$.}

Next, we test the performance of the sampling algorithms in a $d$-dimensional bi-modal Gaussian {mixture} distribution with the two modes centered at ${\mu_1=(-5,\ldots,-5)}$ and $\mu_2= -\mu_1$ and two equal covariance matrices $\Sigma_1=\Sigma_2$ randomly generated as in the previous example. 
The scaling parameters are the same as in the previous example.  Note that the Fisher information in not explicitly known for Gaussian mixtures, hence the Hessian is the only metric we use for smTMCMC. To calculate the errors, we first assign each sample to either one of the modes depending on the shortest Euclidean distance.
{The total error is defined as the  average of the local error on each mode, as defined in \cref{eq:def:err}.}

 \begin{figure}[tbhp]
\centering 
\includegraphics[width=0.5\textwidth]{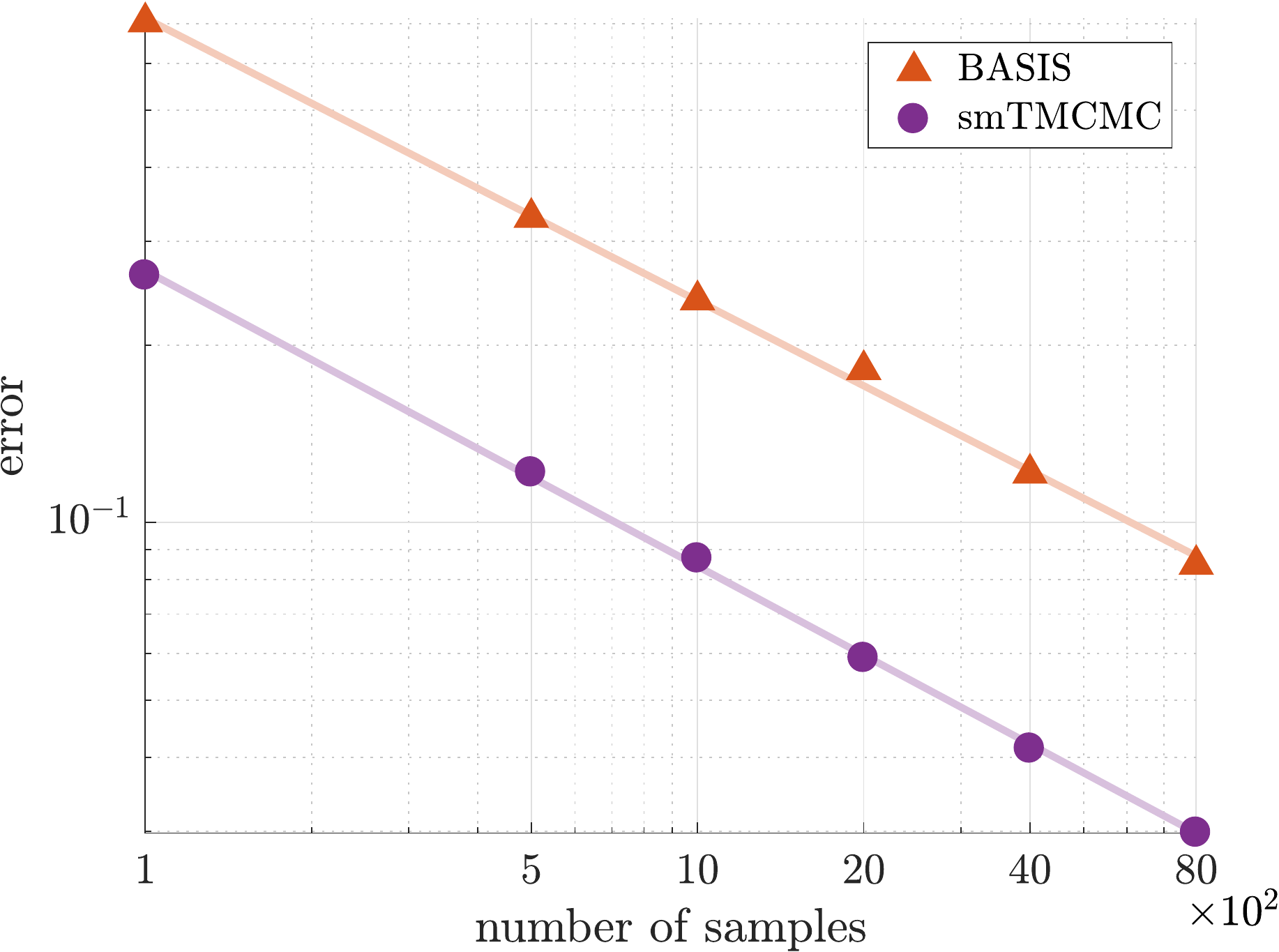} 
\caption{{Averaged error over 100 independent samplings of the 5 dimensional zero mean Gaussian distribution discussed in \cref{section:gaussians}. Both fitted lines to the data (continuous lines) have slope approximately $-\frac{1}{2}$.}}
\label{fig:gaussian:err:vs:dim}
\end{figure}

In  \cref{fig:gaussian:err:b} the error $E$ averaged over 100 independent simulations is presented. The size of the error bars is comparable to that of the markers and thus not included in the graph. The sample size used in this example is chosen to be $N_s=5000$. The reason it is increased compared to the unimodal test case is that TMCMC is not able to detect both modes of the bimodal distribution with less samples. Moreover, with the number of samples fixed, TMCMC is able to detect both modes up to $d=8$ while smTMCMC can go up to 10. Lastly, the difference between the TMCMC and smTMCMC error is increased as a function of dimension.

 The computation of the Hessian or Fisher information comes with an additional cost, but during our experiments no remarkable runtime differences have been observed due to the small computational intensity of the problem. The maximum relative runtime difference has been detected in high dimensions where smTMCMC runs approximately 10\% slower that TMCMC.

A feature of the smTMCMC algorithm is that the number of corrections in the covariance matrix, {i.e., the number of times a matrix $G$ was transformed to a matrix $\widehat{G}$} as discussed in  \cref{sec:covariance}, drop from a very high value (approximately 90\%) for early stages to almost zero during the last stages of the algorithm. This is expected, since during the early TMCMC stages the target distribution is close to uniform. This leads to close to zero values for the Hessian or the Fisher information matrices which in turn leads to  quite wide  proposal distributions that have to be corrected. As TMCMC evolves, the sampled distribution is getting closer to the target distribution and the Hessian or Fisher information has the potential to get improved, at least for locally log-concave functions.

\subsection{A pharmacodynamics model} \label{sec:pharma}

We {apply the proposed algorithms to the  Bayesian inference of a model for the growth for adult, diffuse, low-grade gliomas and their  drug induced inhibition.} We employ the model proposed in \cite{Ribba2012} and infer its  parameters using {clinical} data from MRI measurements of tumors from different patients.  The tumor is composed of proliferative tissue $(P)$ and quiescent tissue $(Q)$. The treatment, due to the drug administration $(C)$ aims to  destroy proliferative cells or transform the quiescent cells to damaged quiescent cells $(Q_P)$. The damaged quiescent cells may either die or repair their DNA and transform into proliferative cells. The model consists of a system of ordinary differential equations,
\begin{equation}
\begin{aligned}
\frac{d C}{dt} &= - \varphi_1 C \\
\frac{d P}{dt} &= \varphi_4 P (1-\frac{P+Q+Q_P}{K}) +  \varphi_5 Q_P - \varphi_3 P - \varphi_1\varphi_2 CP \\
\frac{d Q}{dt} &=  \varphi_3 P - \varphi_1 \varphi_2 C  Q \\
\frac{d Q_P}{dt} &= \varphi_1\varphi_2 C Q - \varphi_5 Q_P - \varphi_6 Q_P \\
C(0) &= 0,  \; P(0)=\varphi_7,  \; Q(0)=\varphi_8, \; Q_P(0) = 0  \PERIOD
\label{eq:pharma}
\end{aligned}
\end{equation}
We set $Y=(C,P,Q,Q_P)^\top$ and the quantity of interest is $ {f(Y,t;\varphi)} = P(t;\varphi) + Q(t;\varphi) + Q_P(t;\varphi)$.
Here the correspondence of the parameter vector $\varphi$ to the parameters used in \cite{Ribba2012} is   
\begin{equation*}
\varphi = (KDE, \gamma, k_{PQ}, \lambda_P, k_{Q_P P} , \delta_{QP},P_0, Q_0) \PERIOD
\end{equation*} 

We process a set of 5 measurements $\DATA_k = \{ (t_i,d_i)  : \, i=1,\ldots,N_d\}$, $k=1,\ldots,5$, that correspond to the measurements of the mean tumor diameter $d_i$ at time instances $t_i$  of 5 patients, and a set of time instances of drug administration $\{ \tau_i : \, i=1,\ldots,N_\tau \}$ is provided \cite{Ribba2012}. At $t=\tau_i$ we restart the simulation and set $C(\tau_i)=1$. {The differential equation system is solved using the Matlab function \textit{ode45}.}
In order to reduce the number of inferred parameters we assume that at time $t=0$ a measurement, $d_1$, is {exactly known} and thus $\varphi_8=d_1-\varphi_7$.

The likelihood function is based on the assumption \cref{eq:likelihood:assumption} of a normally distributed model error and the augmented parameter vector is defined as $\vartheta=(\varphi^\top,\sigma_n)^\top$. 
We choose a uniform prior distribution for the parameters, $p(\vartheta) = \prod_{i=1}^8 \UNIFORM(\vartheta_i | a_i,b_i)$ with $a=( 10^{-2}, 10^{-2}, 10^{-5} , 10^{-5},  10^{-5}, 10^{-5},10^{-5},10^{-5})$ and $b=(20,20,2.5,0.3,0.05,0.6, \allowbreak1,33)$.

Since the prior distribution is uniform, the maximum a-posteriori probability (MAP) estimate coincides with the maximum likelihood (ML) estimate,
\begin{equation}
\vartheta_{ML} =  \argmax_{ {\vartheta} } \, p(\vartheta | \DATA) = \argmax_{ {\vartheta} } \, p(\DATA |  \vartheta) \COMMA
\end{equation}
{where the maximum is taken over all components of $\vartheta \,\,({\vartheta_j>0}, \,\,j=1,\ldots, 8)$}. We optimize the log-likelihood for the different data sets, corresponding to different patients, using the  Covariance Matrix Adaptation Evolution Strategy (CMA-ES) \cite{Hansen2003}. Then, we draw samples from the posterior distribution and compare the maximum log-likelihood in the set of samples with the ML estimate obtained using CMA-ES. {We use this test as an indicator to check whether the sampling algorithm is able to sample the high probability area close to the ML point.}

In \cref{table:map} we report the estimated ML  by the CMA algorithm and compare with the maximum log-likelihood in the sample set of TMCMC and smTMCMC algorithms. The results of pTMCMC are indistinguishable from those of smTMCMC and thus not reported here. The TMCMC algorithm is unable to identify the ML while the smTMCMC gets closer to maximum likelihood.
Increasing the sample size consistently improves the results of smTMCMC while this is not true for TMCMC, see for example the decrease in the log-likelihood for $10^5$ samples for patients 3 and 5 in   \cref{table:map}.

The difficulty that the sampling algorithms faces in trying to identify the high probability areas, suggests that an unidentifiable manifold is present in the posterior probability space. 
Using the profile log-likelihood (PL) function \cite{Raue2012},
\begin{equation} \label{eq:PL}
\mathrm{PL}(\vartheta_i | \DATA) =  \max_{\vartheta_{\backslash i}} \, p(\vartheta | \DATA)  \COMMA
\end{equation}
{where the notation $\vartheta_{\backslash i}$ implies that we fix the $i$-th element and we maximize over the remaining elements of $\vartheta$.}
We verified this assumption for  $i=1$. Areas of constant PL indicate non identifiable parameters and as shown in  \cref{fig:PL:pat1} the PL function exhibits  a large area of unidentifiability for $\vartheta_1 > 2$.

\begin{table}[htpb]
\begin{center}
\begin{tabular}{ | M{0.5cm} || M{1.8cm} || M{1.8cm} | M{2cm} || M{2cm} |  M{2cm} | N} \hline
No & CMA-ES & TMCMC $N_s=1e4$ & TMCMC $N_s=1e5$ & smTMCMC $N_s=1e4$ & smTMCMC $N_s=1e5$ &   \\[5pt] \hline
1& -33.06 &	-38.81&	-36.99&	-34.89& -33.66  & \\[5pt] \hline
2&  -29.91&	-39.94&	-35.67&	-32.11&	-30.96 &	\\[5pt] \hline
3&  -26.24&	-28.11&	-30.39&	-26.79&	-26.40 &	\\[5pt] \hline
4&  -26.69&	-29.85&	-28.18&	-27.39& -26.88  & \\[5pt] \hline
5&  -15.68&	-35.41&	-36.23&	-19.42&	-17.87 &	\\[5pt] \hline
\end{tabular}
\end{center}
\caption{The maximum log-likelihood estimate for the Pharmacodynamics model \cref{eq:pharma} obtained using the CMA-ES algorithm, compared with the maximum log-likelihood found in the sample set using the TMCMC and the  smTMCMC algorithms.}
\label{table:map}
\end{table}

In  \cref{fig:TMCMC:pat1}, \cref{fig:mTMCMC:pat1:1e4} and \cref{fig:mTMCMC:pat1:1e5} we present the posterior samples conditioned on  data from the first patient. In \cref{fig:TMCMC:pat1} the sampling was done using the TMCMC algorithm and $10^4$ samples, in  \cref{fig:mTMCMC:pat1:1e4} and  \cref{fig:mTMCMC:pat1:1e5} using the smTMCMC and $10^4$ and $10^5$ number of samples, respectively. The histograms for the 8 parameters are plotted along the diagonal. Pair samples and a smoothed version of the pair marginal histogram  are plotted in the  upper and lower triangular part of the figure, respectively. We note  that the TMCMC is being trapped in local maxima of the posterior distribution in all directions and it is not able to correctly populate the posterior sample space. This is evident from the spikes on the histograms on  \cref{fig:TMCMC:pat1}. This unnatural local mass concentration does not appear in the smTMCMC sample set shown in   \cref{fig:mTMCMC:pat1:1e4}.

We have also observed that the TMCMC algorithm does not produce consistent samples. The shape of the estimated distribution varies significantly between individual runs of the algorithm. Moreover, this problem is not  fixed by increasing the size of the sample set. In contrast, this is not the case for the smTMCMC algorithm. The distribution presented in  \cref{fig:mTMCMC:pat1:1e5} does not change between individual runs of the algorithm or if  the number of samples is increased. This becomes evident by comparing  \cref{fig:mTMCMC:pat1:1e4} and \cref{fig:mTMCMC:pat1:1e5} where $10^4$ and $10^5$ samples were used, respectively.

The location of the ML estimate using CMA-ES and  the smTMCMC algorithm is marked in the diagonal histograms of  \cref{fig:mTMCMC:pat1:1e4} and \cref{fig:mTMCMC:pat1:1e5} with $\times$ and $\circ$, respectively. It can be observed that there is a discrepancy from the CMA-ES estimate in  \cref{fig:mTMCMC:pat1:1e4}, where $10^4$ samples were used.  In  \cref{fig:mTMCMC:pat1:1e5} the discrepancy is alleviated by increasing the sample size to $10^5$. With this observations, it becomes evident that the high probability areas of the posterior distribution are  populated.
In summary, we find that the smTMCMC, encompasses the properties of TMCMC and its extensions enable it to  effectively sample the posterior distribution of a challenging Pharmacodynamics model  using clinical data.

\section{Conclusion}

We have proposed a new population based sampling algorithm that combines the Transitional MCMC (TMCMC) algorithm with Langevin diffusion transition kernels. Instead of using isotropic diffusion in the TMCMC algorithm, the proposal scheme is  based on the time discretization of a Langevin diffusion that takes into consideration the geometry of the target distribution. Thus diffusion is taking place in a manifold that is defined either through the Hessian or the Fisher information matrix of the distribution.  

The adaptation to the local geometry of the distribution enhances the  sampling capabilities of the method. At the same time the requirement for positive definite Hessian of Fisher information matrix can not be guaranteed for general distributions. Even if the matrices are positive definite  some eigenvalues may be close to zero, severely affecting the efficiency of the algorithm. Such cases of nearly zero eigenvalues is present in systems with unidentifiable parameters.
We alleviate the problems associated with these matrices by introducing the  manifold TMCMC (mTMCMC). We successfully applied the proposed algorithm  using a Gaussian distribution with probability mass concentrated in the boundaries of the prior distribution as well as  with a {bimodal Gaussian mixture} distribution, greatly increasing  the sampling quality. Finally, we showcased  the ability of the proposed algorithm to sample challenging, multimodal distributions in the presence of unidentifiable manifolds using the posterior distribution in a Pharmacodynamics model, and compared its efficiency with respect to existing sampling methods.

The extension of the algorithm for priors other than uniform  is also conditionally possible. If the prior distribution is not log-concave then the proposed algorithm is not applicable. In the case of a log-concave  prior distribution, one has  to redefine the extended boundary technique, presented in  \cref{sec:covariance}, in the absence of a prior with compact support.

\begin{figure}
\begin{center}
	\includegraphics[width=1\textwidth]{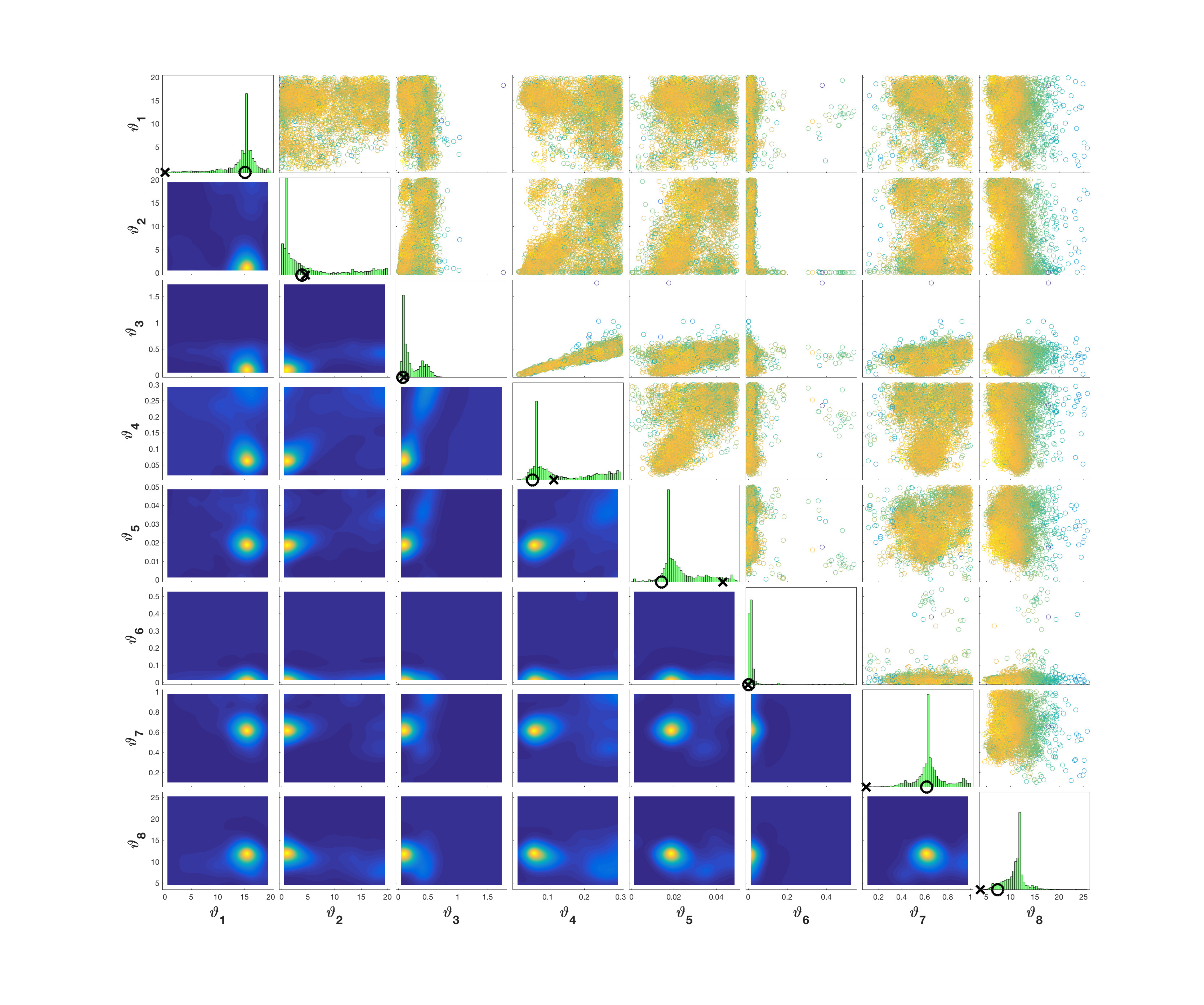}	
\end{center}
	\caption{Samples form the posterior distribution of Pharmacodynamics model \cref{eq:pharma} conditioned on  data from patient 1 and using the TMCMC algorithm {with $10^4$ samples}. The $\times$ and $\circ$ correspond to the ML estimates obtained by the CMA-ES and the TMCMC algorithm, repsectively.}	
    \label{fig:TMCMC:pat1}
\end{figure}

\begin{figure}
\centering
	\includegraphics[width=\textwidth]{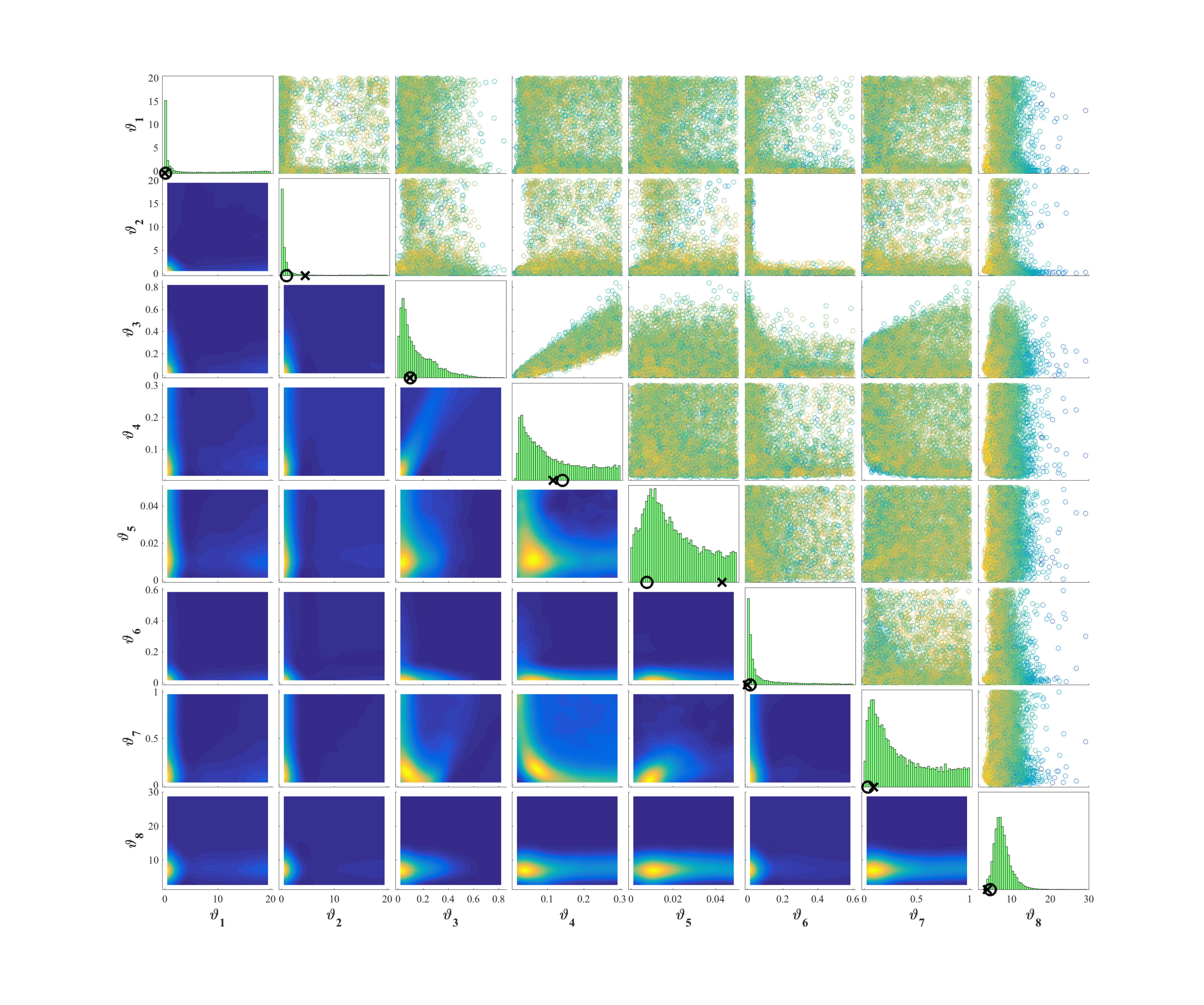}	
	\caption{Samples from the posterior distribution of Pharmacodynamics model \cref{eq:pharma} conditioned on  data from patient 1 and using the smTMCMC algorithm with $10^4$ samples. The $\times$ and $\circ$ correspond to the ML estimates obtained by the CMA-ES and the smTMCMC algorithm, repsectively. }	
    \label{fig:mTMCMC:pat1:1e4}
\end{figure}

\begin{figure}
\centering
	\includegraphics[width=\textwidth]{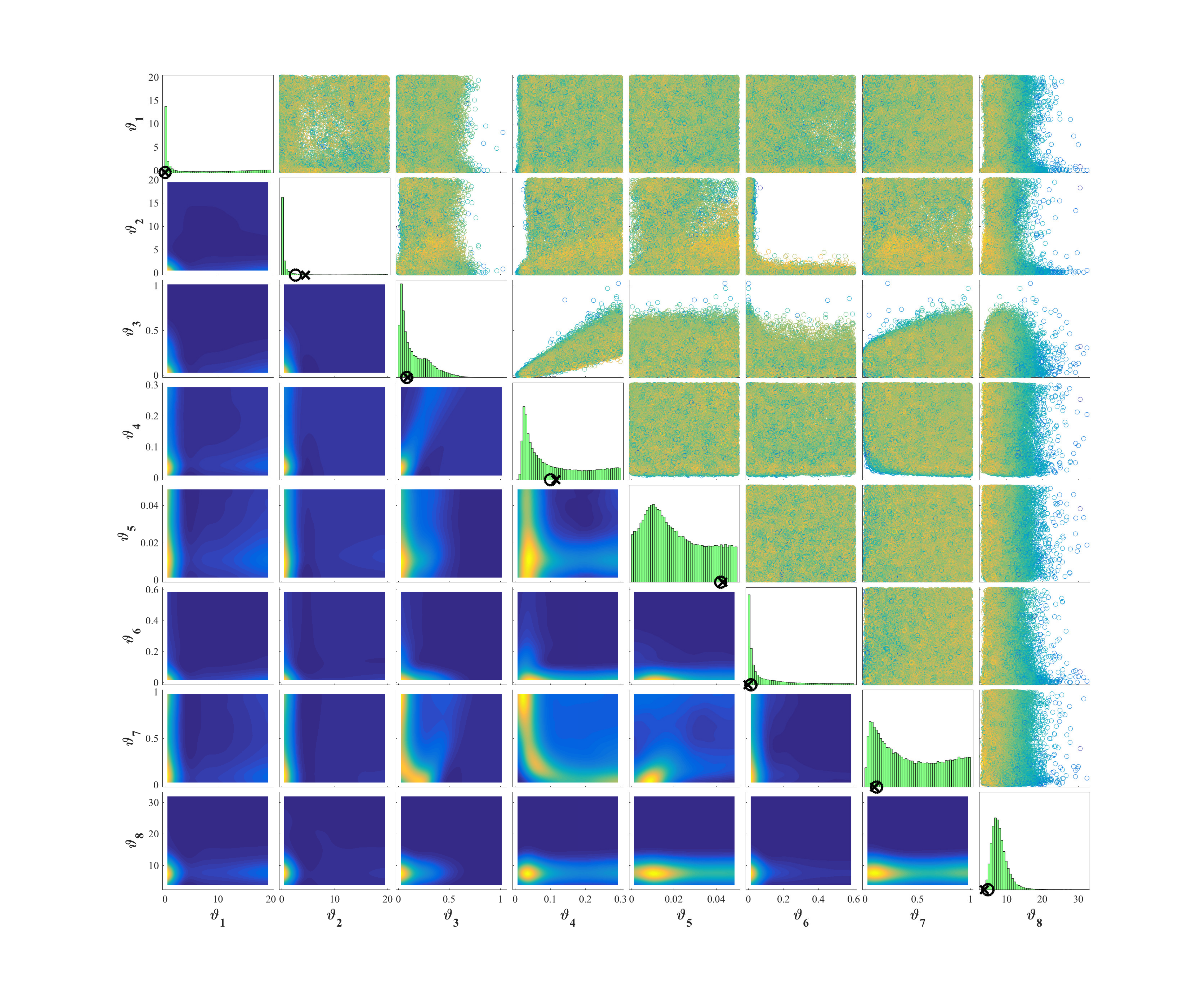}	
	\caption{Samples form the posterior distribution of Pharmacodynamics model \cref{eq:pharma} conditioned on  data from patient 1 and using the smTMCMC algorithm with $10^5$ samples. The $\times$ and $\circ$ correspond to the ML estimates obtained by the CMA-ES and the smTMCMC algorithm, repsectively. }	
    \label{fig:mTMCMC:pat1:1e5}
\end{figure}

\begin{figure}
\centering
	\includegraphics[width=0.5\textwidth]{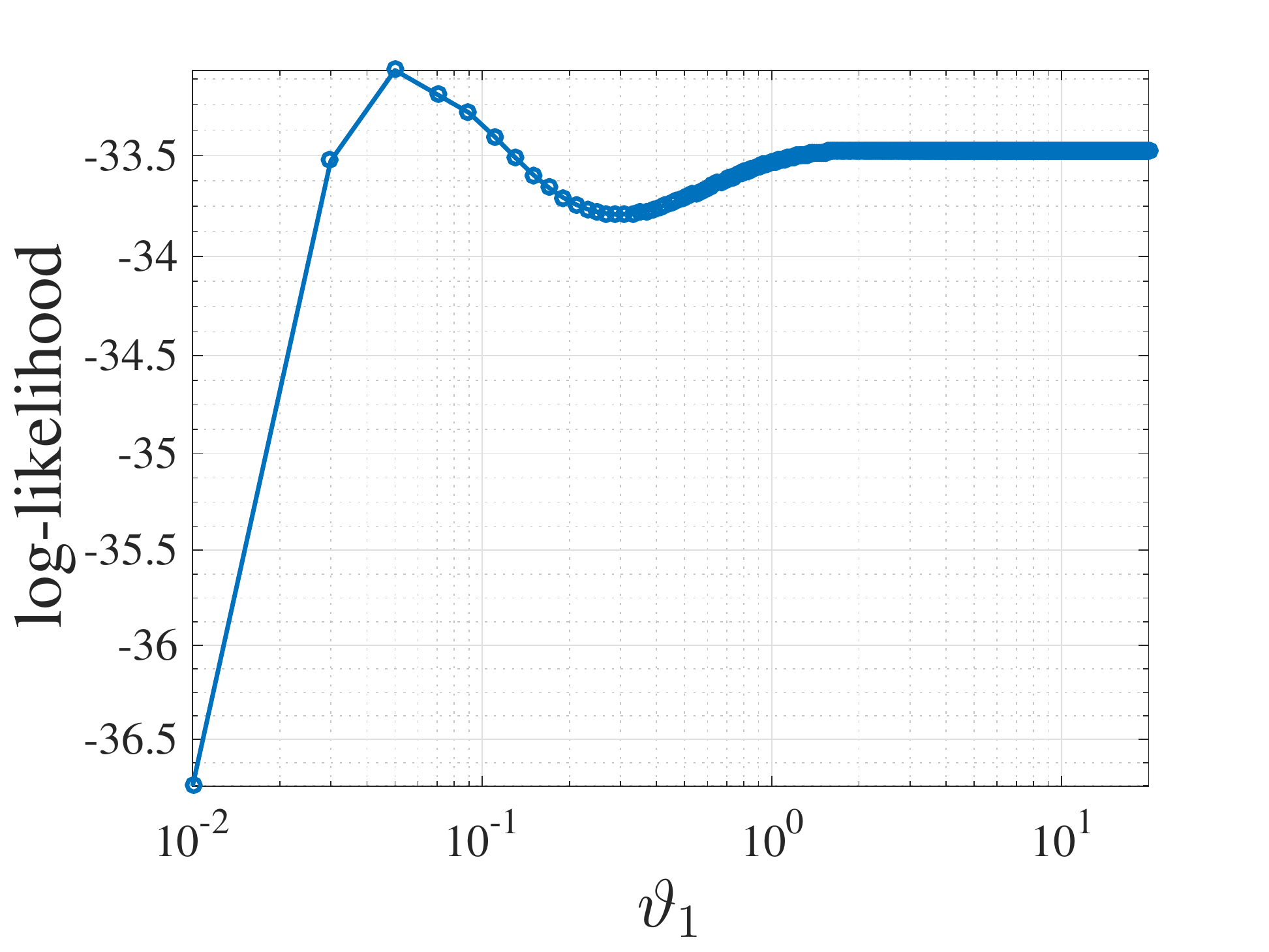}	
	\caption{The profile log-likelihood function \cref{eq:PL} for the $\vartheta_1$ parameter for the Pharmacodynamics model \cref{eq:pharma} conditioned on  data from patient 1. }
    \label{fig:PL:pat1}
\end{figure}

\section*{Acknowledgments}
We would like to thank Dr Benjamin Ribba  and Dr Francois Ducray for permission to use the data in the Pharmacodynamics model of  \cref{sec:pharma}. We would also like to thank the anonymous reviewers for their valuable comments. Finally, we would like to acknowledge the computational time at Swiss National Supercomputing Center (CSCS) under the project s659 and funding support from the European Research Council (Advanced Investigator Award no. 341117).

\newpage

\bibliographystyle{abbrv}
\bibliography{bibliography}

\newpage

\appendix

\section{Derivatives of Ordinary Differential Equations}\label{sec:ODEs}

In this section we show how the derivatives of the model output, $f$, with respect to the model parameters, $\varphi$, can be computed. These quantities appear in the evaluation of the derivative of the log-likelihood function with respect to the parameters, see \cref{eq:loglik:d1} and \cref{eq:loglik:d2}.

Let {$f:=f(Y,t;\varphi)$} be an observable function on the solution of the ODE system,
\begin{equation}
\begin{cases}
\dot Y(t) = G(Y;\varphi), \quad t\in[0,T] \\
Y(0) = g(\varphi)     \COMMA
\end{cases}
\label{eq:gen:ode}
\end{equation}
where $Y\in\REAL^{N_Y}$ and $\varphi\in\REAL^{N_{\varphi}}$ the vector of parameters. Here, we assume that $G$ is a smooth function such that all the derivatives used bellow are well defined. In  order to find the first and second order derivatives needed for the manifold algorithms, the following extended system must be solved,
\begin{equation}
\label{eq:extended:ode}
\begin{cases}
\dot Y = G(Y;\varphi) \\
\dot S_k = G_{k}^1(Y,S;\varphi), \quad  k=1,\ldots,N_\varphi\\
\dot H_{k,\ell} = G_{k,\ell}^2 (Y,S,H;\varphi), \quad  k,\ell=1,\ldots,N_\varphi \\
Y(0) = g(\varphi), \;\; S(0) = g_{k}^1(\varphi),  \;\; H(0) = g^2_{k,\ell}(\varphi) \COMMA
\end{cases}
\end{equation}
where $G_{k}^1$ and $G_{k,\ell}^2$ is the first total derivative of $G$ with respect to $\varphi_k$ and second total derivative of $G$ with respect to $\varphi_k$ and $\varphi_\ell$, respectively, and  $g^1_k  = \PARTIAL{\varphi_k} g$ and $g^2_{k,\ell}  = \PARTIALTWO{\varphi_k}{\varphi_\ell} g$.

The function $G_{k}^1 = D_{\varphi_k} G $ for $ k=1,\ldots,N_\varphi $ is given by, 
\begin{equation}
G_{k}^1 = A S_k + B_k \COMMA
\end{equation}
where 
\begin{equation}\label{eq:ode:a}
A_{ij} := A_{ij}(Y;\varphi)  = \frac{\partial}{\partial Y_j} G_i(Y;\varphi)    \COMMA
\end{equation}
for $i,j=1,\ldots,N_Y$ and
\begin{equation}
B_k := B_k(Y;\varphi) =  \frac{\partial}{\partial \varphi_k} G(Y;\varphi) \PERIOD
\end{equation}

The function $G_{k,\ell}^2 = D_{\varphi_k} D_{\varphi_\ell} G $ for $k,\ell=1,\ldots,N_\varphi$ is given by, 
\begin{equation}
G_{k,\ell}^2 =   AH_{k,\ell} + ( I \otimes S_k^\top) C_{k,\ell} (\ONE\otimes S_\ell) +  D_k S_\ell + D_\ell S_k + J_{k,\ell} \COMMA
\end{equation}
where $A$ is defined in \cref{eq:ode:a}, $\otimes$ is the Kronecker product, $I\in\REAL^{N_Y\times N_Y}$ the identity matrix and $\ONE=(1,\ldots,1)\in\REAL^{N_Y}$. The matrix $C$ is a block diagonal matrix  $C_{k,\ell} = \DIAG(C_{k,\ell,1},\ldots,C_{k,\ell,N_Y})$ with the block matrices given by
\begin{equation}
C_{k,\ell,i} := C_{k,\ell,i}(Y;\varphi) = \frac{\partial^2}{\partial Y_k \partial Y_\ell} G_i(Y;\varphi) \PERIOD
\end{equation}
The matrix $D_k$ and the vector $J_{k,\ell}$ are given by
\begin{equation}
D_{k,i,j} := D_{k,i,j}(Y;\varphi) =  \frac{\partial^2}{\partial \varphi_k  \partial Y_j } G_i(Y;\varphi) \COMMA
\end{equation}
\begin{equation}
J_{k,\ell} := J_{k,\ell}(Y;\varphi) =  \frac{\partial^2}{\partial \varphi_k \partial \varphi_\ell} G(Y;\varphi) \COMMA
\end{equation}
for $i,j=1,\ldots,N_Y$.

A Matlab function is provided which given  functions $G$ and $g$, by performing symbolic calculations, gives as output the functions $G_{k}^1$ and $G_{k,\ell}^2$ as Matlab function handles.

\section{Derivative of the transformed matrix} \label{appendix:derivative:G}

\begin{proposition}
Let $A\in\REAL^{n\times n}$ be a matrix that depends on a parameter $\vartheta$ and let $A=Q\Lambda Q^\top$ its eigendecomposition. Let $f(A) = Q f(\Lambda) Q^\top$ be a transformation of the matrix $A$. Then, it holds that
\begin{equation}
\frac{\partial f(A) }{\partial \vartheta}  = Q \Big (     J \circ      \big( Q^\top \, \frac{\partial A }{\partial \vartheta} \, Q \big )   \Big  ) Q^\top \COMMA
\end{equation}
where $\circ$ denotes the Hadamard product and $J$ is given by,
\begin{equation}
J_{i,j} = 
\begin{cases}
\frac{ f(\lambda_i) - f(\lambda_j) }{  \lambda_i - \lambda_j}, & i\neq j \\
\frac{\partial f(\lambda_i) }{\partial \lambda_i}, & i=j \PERIOD
\end{cases}
\end{equation}
\end{proposition}

\begin{proof}
For the proof see Section 2 in \cite{Aizu1963}.
\end{proof}

\end{document}